\documentclass[11pt, reqno]{amsart}

\makeatletter
\g@addto@macro{\endabstract}{\@setabstract}
\makeatother

\usepackage{graphics, stackrel}
\usepackage{amsmath, amssymb, amsthm}
\usepackage{amsfonts}
\usepackage{graphicx}
\usepackage{verbatim}
\usepackage{natbib}

\usepackage{fancyvrb}
\usepackage{color}
\usepackage{mdwlist}

\usepackage[citecolor=blue, colorlinks=true, linkcolor=blue]{hyperref}

\usepackage{mathrsfs}
\usepackage{bbm}

\usepackage[left=1.35in, right=1.35in, top=1.0in, bottom=1.15in, includehead, includefoot]{geometry}

\usepackage{enumitem}
\setlist[enumerate]{itemsep=2pt,topsep=3pt}
\setlist[itemize]{itemsep=2pt,topsep=3pt}
\setlist[enumerate,1]{label=(\alph*)}

\usepackage{multirow}

\renewcommand{\leq}{\leqslant}
\renewcommand{\geq}{\geqslant}

\DeclareMathOperator*{\argmax}{argmax}
\DeclareMathOperator*{\argmin}{argmin}

\newcommand{\setntn}[2]{ \{ #1 : #2 \} }

\newcommand{\1}{\mathbbm 1}

\newcommand*\diff{\mathop{}\!\mathrm{d}}
\renewcommand{\epsilon}{\varepsilon}

\newcommand{\cC}{\mathscr C}

\newcommand{\vV}{\mathscr V}

\newcommand{\rR}{\mathcal R}

\newcommand{\XX}{\mathsf X}
\renewcommand{\AA}{\mathsf A}

\newcommand{\ZZ}{\mathsf Z}
\renewcommand{\SS}{\mathsf S}

\newcommand{\GG}{\mathbbm G}

\newcommand{\RR}{\mathbbm R}

\newcommand{\NN}{\mathbbm N}

\newcommand{\EE}{\mathbbm E \,}

\theoremstyle{plain}
\newtheorem{theorem}{Theorem}[section]

\newtheorem{lemma}[theorem]{Lemma}
\newtheorem{proposition}[theorem]{Proposition}

\theoremstyle{definition}

\newtheorem{assumption}{Assumption}[section]

\begin{document}

\title{}

\date{\today}

\begin{center}
  \Large
  Dynamic Programming with Recursive Preferences: \\ Optimality and
  Applications\footnote{This paper has benefited from thoughtful comments by
      Jaroslav Borovi\v{c}ka, Yiannis Vailakis and Gaetano Bloise.  
  Financial support from ARC grant FT160100423 is gratefully acknowledged.
  Email addresses: \texttt{guanlong.ren@anu.edu.au} and \texttt{john.stachurski@anu.edu.au} }

  \bigskip
  \normalsize
  Guanlong Ren and John Stachurski \par \bigskip

  Research School of Economics, Australian National University \bigskip

  \today
\end{center}

\begin{abstract} 
    This paper provides new conditions for dynamic optimality in discrete time and uses them to establish fundamental dynamic programming results for several commonly used recursive preference specifications.   These include Epstein-Zin preferences, risk-sensitive preferences, narrow framing models and recursive preferences with sensitivity to ambiguity.  The results obtained for these applications include (i) existence of optimal policies, (ii) uniqueness of solutions to the Bellman equation, (iii) a complete characterization of optimal policies via Bellman's principle of optimality, and (iv) a globally convergent method of computation via value function iteration.  

    \vspace{1em}
\end{abstract}


\section{Introduction}

Economists have constructed progressively more realistic representations of
agents and their choices.  In the context of intertemporal choice, these
preferences have come to include such features as independent sensitivity to
intertemporal substitution and intratemporal risk, desire for robustness, the
impact of narrow framing and some forms of ambiguity.  Models with such
features now play a key role in connecting quantitative
models to data in several areas of macroeconomics and finance.\footnote{For example,
    the recursive preference specification of
    \cite{epstein1989} now forms a core component of the quantitative asset
    pricing literature, while also finding use in applications ranging from
    optimal taxation to fiscal policy and business cycles (see, e.g.,
    \cite{bansal2004yaron}, \cite{kaplan2014model},
    or \cite{schorfheide2018identifying}).  
    For a canonical study of desire for robustness in economic modeling see
    \cite{hansen2008robustness}.  Well known studies of ambiguity and
    intertemporal choice include
    \cite{epstein2008ambiguity},
    \cite{klibanoff2009recursive},
    \cite{siniscalchi2011dynamic}
    and \cite{ju2012ambiguity}.  An overview of recursive preference models can be found in \cite{backus2004exotic}.}

At the same time, uptake of sophisticated intertemporal preferences is
hampered by technical difficulties, which are in turn a result of the
nonlinearities embedded in such preferences.  Moreover, even for those
settings where recursive preferences \emph{are} routinely adopted, such as
quantitative models of asset prices and equity premia, it has become clear
that the extensive approximations commonly used to simplify them can generate large and
economically significant errors (see, e.g., \cite{pohl2018higher}).
 
There are two major technical difficulties associated with recursive
preference models.  One is existence and
uniqueness of utility when the latter is recursively defined.  For example, if
lifetime utility $U_t$ at time $t$ is specified as 
\begin{equation} \label{eq:egru} U_t = W(C_t, \rR (U_{t+1}))   ,
\end{equation}
where $W$ is a some form of aggregator, $\rR$ is a certainty equivalent
operator and $\{C_t\}$ is a given consumption path (see, e.g., \cite{epstein1989}), then is 
utility $\{U_t\}$ well defined?  Is it uniquely defined?  What conditions are required on $W$ and
$\{C_t\}$ such that the answers are affirmative?

After a period of extensive research, these issues have
largely been settled for specifications in common use.\footnote{See,
    for example, \cite{epstein1989}, \cite{rincon2007recursive},
    \cite{klibanoff2009recursive}, \cite{marinacci2010unique},
\cite{hayashi2011intertemporal}, \cite{hansen2012recursive},
\cite{martins2013fixed}, \cite{becker2017rzapa},
\cite{borovioka2017stachurski} or \cite{guo2019he}.} However,
there remains the second significant problem associated with recursive
utilities: how to \emph{maximize} them.  For example, if consumption $\{C_t\}$
in \eqref{eq:egru} is one of many time paths associated with different
household savings and investment policies, then which feasible policy will
maximize $U_t$?  Does such a policy exist and, if so, how can we obtain it?
More generally, how can we bring the theory of dynamic programming to bear on
optimal choice when preferences are defined recursively?  

In this paper, we develop a theory of dynamic programming suitable for
recursive preference models that is broad, simple to state, practical and highly applicable.  We then use this theory to obtain fundamental
optimality results for several common specifications in the literature,
including empirically relevant parameterizations of Epstein--Zin models with
constant elasticity of substitution aggregators, risk-sensitive and robust
control models, models with narrow framing and ambiguity sensitive
preferences.  In each case we show that the value function is the unique
solution to the Bellman equation, that value function iteration converges
uniformly to the value function, that Bellman's principle of optimality is
valid, and that an optimal policy exists.

The novelty of the approach to dynamic programming taken in this paper is that
we use monotone convex operators to study the maximization problems associated
with dynamic programming.  While monotone convex operators have received
little attention in the economics literature to date (unlike, say, contraction
mappings or monotone \emph{concave} operators), they have a significant
advantage in the present context: \emph{convexity is preserved under the taking of
pointwise suprema.}  We exploit this heavily in our results, along with the
fact that, under suitable conditions, monotone convex operators enjoy
attractive stability properties.

We combine these ideas with the abstract dynamic programming framework of
\cite{bertsekas2013abstract}, who obtains strong optimality results in a
general setting under a (weighted) contraction condition.
This latter condition fails for most of the recursive preference models used
routinely in economics and finance.
Here, we maintain the abstract dynamic programming framework
used by \cite{bertsekas2013abstract}, along with a fundamental monotonicity
condition, but replace his contraction condition with a convexity condition.
The latter condition (i) turns out to be satisfied by many recursive preference
models, and (ii) can be used to connect with the theory of monotone convex
operators discussed above, replacing Banach's contraction mapping theory and
delivering the foundational results of dynamic programming (e.g., convergence of
value function iteration, validity of Bellman's principle of optimality, and
existence of an optimal policy).

As well as the theory and applications discussed above, we also provide two
extensions of our results.  First, in addition to
maximization problems, we treat minimization problems in the same abstract
dynamic programming framework.  Second, we show how the framework can be
extended to handle unbounded rewards, providing fundamental dynamic
programming results for an
Epstein--Zin recursive preference model with noncompact state space.  

Our work is inspired by several existing studies, in addition to those
mentioned above.  Two of the most closely related are \cite{bloise2018convex}
and \cite{marinacci2019unique}, both of which analyze dynamic programming
in the setting of recursive utility (in addition to making other contributions
such as general fixed point theory for monotone operators and properties of
Koopmans operators).  In both cases, the authors exploit monotonicity and
concavity properties inherent in some recursive preferences to
obtain conditions under which the Bellman operator has a unique and globally
attracting solution within a given class.

Our treatment of foundational results in dynamic programming draws on and
extends their theoretical analysis in several directions, under a set of
assumptions that is related but not directly comparable.  In particular, in
addition to providing conditions under which the Bellman operator has a unique
and globally attracting solution within a given class, as in
\cite{bloise2018convex} and \cite{marinacci2019unique}, we also show that,
under the same conditions, that solution is the value function and Bellman's
principle of optimality is valid, in the sense that a policy is optimal if and
only if it is obtained by taking the maximizing action at each state on the
right hand side of the Bellman equation.  Furthermore, we show that at least one such
policy exists.

On the practical side, we extend \cite{bloise2018convex} and
\cite{marinacci2019unique} by providing a detailed set of dynamic programming
results for Markov decision processes with (i) Epstein--Zin preferences, (ii)
risk-sensitive preferences, (iii) recursive smooth ambiguity preferences and
(iv) narrow framing.
These are among the most common formulations of recursive preferences in the
applied literature.  


Our treatment of dynamic programming in the setting of Epstein--Zin
preferences, contained in section~\ref{ss:ezp-app} below, adds value to the
original work of \cite{epstein1989}, as well as an extension by
\cite{guo2019he}.  For example, while these dynamic programming results treat
specific consumption-savings and portfolio selection problems with no
endowments, taxes, labor income, or other forms of non-financial income (see
\cite{epstein1989}, theorem~5.1, and \cite{guo2019he}, section~5), we treat a
generic Markov decision problem.  This facilitates adopting Epstein--Zin
preferences in models with features such as optional default, transaction
costs or occasionally binding constraints.  

Our treatment of risk-sensitive preferences complements recent work in
\cite{bauerle2018stochastic}, who consider a stochastic optimal growth problem
in that setting.  We
replace the one-sector optimal growth model with a generic Markov decision
problem, modified to allow for risk-sensitivity with respect to continuation
values.  This result can be applied directly to a large range of dynamic
programs, which is significant partly because, in many cases, the Bellman
equation of a risk-sensitive agent can also be interpreted as the key
functional equation for determining the optimal policy of a robust controller
(see, e.g., \cite{hansen2008robustness}).

Our treatment of dynamic programming in the setting of recursive smooth
ambiguity is based on the preference specification used in
\cite{ju2012ambiguity}. This specification admits a three way separation among
risk aversion, ambiguity aversion and intertemporal substitution, and has
proved valuable in interpreting a variety of asset-pricing phenomena.
In that paper, consumption is given exogenously.
While existence and uniqueness of recursive utility is a solved problem (see,
e.g., \cite{hayashi2011intertemporal}), dynamic programming is not.
Hence, we add to the existing stock of knowledge on this framework an
existence result for optimal policies, a characterization of these policies
via Bellman's principle of optimality, and a global means of computation via
value function iteration.

Our treatment of Epstein--Zin models with unbounded rewards extends previous
work in economics on dynamic programming with additively separable prefernces
as unbounded rewards, as found, for example, in \cite{rincon2003existence}, \newline
\cite{martins2010existence}, \cite{kamihigashi2014elementary} and
\cite{bauerle2018stochastic}.  Our approach is based on weighting functions
and departs from earlier research in that
weighting is combined with a convexity assumption rather than a
weighted supremum norm contraction condition, in order to handle common
recursive preference specifications.

There are several other papers related to the theory of dynamic programming
with recursive preferences that should be mentioned here.  One is
\cite{le2005recursive}, who study optimal growth in the presence of
recursively defined utility, allowing for both bounded and unbounded returns.
They provide detailed optimality results, including existence of optimal
policies and a version of Bellman's principle of optimality, as well as a way
to compute solutions.  On the other hand, as the authors themselves point out,
the weighted contraction conditions they require fail for many standard
preference specifications, including Epstein--Zin utility.

Also related is \cite{balbus2016}, who considers a class of non-negative
recursive utilities with certain types of nonlinear aggregators and certainty
equivalent operators, and studies the corresponding dynamic programming
problem.  His results for existence, uniqueness and convergence of solutions
to the Bellman equation rest upon the theory of monotone $\alpha$-concave
operators.  We omit a detailed discussion because the 
$\alpha$-concave property does not appear to hold for many of the specifications of
recursive utility we consider here, such as Epstein--Zin utility when
the elasticity of intertemporal substitution differs from unity.\footnote{The theory of monotone concave
    operators was previously used in \cite{levan2008morhaim} for additively
    separable problems, and to show existence of Markov equilibria in the
    presence of distortions in \cite{datta2002existence},
\cite{morand2003existence} and several related papers.} 

Similarly, \cite{ozaki1996streufert} provide a comprehensive study of the
recovery of recursive preferences and the corresponding dynamic programming
problem under a non-Markovian environment.  Their results are useful for
studying dynamic programming for non-additive stochastic objectives in a very
general setting, although their assumptions exclude some popular
specifications.  For example, with Epstein--Zin preferences and elasticity of
intertemporal substitution greater than one, the conditions of Theorem~D fail
(since the parameters related to the variable discount factor $\bar{\delta}$
and $\delta$ are infinite).


Our paper also has some connection to the recent work by
\cite{pavoni2018dual}, who introduce a recursive dual approach suitable for
incentive-constrained programming
problems.  The authors construct a dual formulation of the
applications they consider, which is recursive and such that the dual Bellman operator is
contractive under a Thompson-like metric.  Their theory can handle problems where
preferences are specified via a general time aggregator and stochastic
aggregator.  In the problems we consider below, forward looking constraints
are absent and we can directly consider the primal optimization problem.  This
allows us to avoid assumptions used to tie the primal problem to the
dual and obtain contractivity on the dual side.

The remainder of the paper is structured as follows: Section~\ref{s:grs}
contains our main results.  Section~\ref{s:a} has applications.
Section~\ref{s:ubdd} considers an extension (unbounded rewards).  Apart from
some simple arguments, all proofs are deferred to the appendix.

\section{General Results}

\label{s:grs}

Let $\XX$ and $\AA$ be separable metric spaces, called the
\emph{state} and \emph{action space} respectively.
Let $\RR^\XX$ represent all functions from $\XX$ to $\RR$ and
let $\| \cdot \|$ denote the supremum norm on the bounded functions in
$\RR^\XX$.  For $f$ and $g$ in $\RR^\XX$,
the statement $f \leq g$ means $f(x) \leq g(x)$ for all $x \in \XX$.
Let $\Gamma$ be a correspondence from $\XX$ to $\AA$, referred to
below as the \emph{feasible correspondence}.  We understand $\Gamma(x)$ as
representing all actions
available to the controller in state $x$.  The correspondence $\Gamma$ in turn
defines the set of \emph{feasible state-action pairs}
\begin{equation*}
    \GG := \setntn{(x, a) \in \XX \times \AA}{a \in \Gamma(x)}.
\end{equation*}
Let 
\begin{itemize}
    \item $w_1$ and $w_2$ be bounded continuous functions in $\RR^\XX$ satisfying $w_1
        \leq w_2$,
    \item $\vV$ be all Borel measurable functions $v$ in $\RR^\XX$ satisfying $w_1 \leq v \leq w_2$, and
    \item $\cC$ be the continuous functions in $\vV$.  
\end{itemize}
Both $\vV$ and $\cC$ are understood as classes of candidate value functions.  
The functions $w_1$ and $w_2$ serve as lower and upper bounds for
lifetime value respectively. Their role will be clarified below.

Current and future payoffs are subsumed into a \emph{state-action aggregator} $H$, which maps a
feasible state-action pair $(x, a)$ and function $v$ in $\vV$ into a real value $H(x,
a, v)$.  The interpretation of $H(x, a, v)$ is total lifetime rewards, contingent
on current action $a$, current state $x$ and the use of $v$ to evaluate future
states.  In other words, $H(x, a, v)$ corresponds to the right hand side of
the Bellman equation when $v$ represents the value function. 
A simple example is given in Section~\ref{ss:addsep-app}.

The central role of convexity and concavity was discussed in the introduction.
To implement the corresponding restrictions, we call $H$ 
\emph{value-convex} if 
\begin{equation*}
H(x, a, \lambda v + (1-\lambda) w) \leq \lambda H(x, a, v ) + (1-\lambda) H(x, a, w)
\end{equation*}
for each $(x, a) \in \GG$, $\lambda \in [0, 1]$ and $v, w$ in $\vV$.
Similarly, $H$ will be called \emph{value-concave} when the reverse inequality
holds (i.e., when $-H$ is value-convex).  One of these restrictions will be
imposed on each problem we consider.

We also impose some basic properties that will be assumed in every case:

\begin{assumption}
    \label{a:ath}
    The following conditions hold:
    \begin{enumerate}
        \item The feasible correspondence $\Gamma$ is nonempty, compact valued and
            continuous.
        \item The map $(x, a) \mapsto H(x, a, v)$ is Borel measurable on $\GG$
            whenever $v \in \vV$ and continuous on $\GG$ whenever $v \in \cC$.
        \item The state-action aggregator satisfies 
            \begin{equation}
                \label{eq:mon}
                v \leq v' \implies
                H(x, a, v) \leq H(x, a, v')
                \quad \text{for all } (x, a) \in \GG.
            \end{equation}
        \item The functions $w_1$ and $w_2$ satisfy
            \begin{equation}
                \label{eq:uls}
                w_1(x) \leq H(x, a, w_1)
                \quad \text{and} \quad
                H(x, a, w_2) \leq w_2(x)
            \end{equation}
            for all $(x, a)$ in $\GG$.
    \end{enumerate}
\end{assumption}

The primary role of (a) and (b) is to obtain existence of
solutions.  If the state and action space are discrete (finite or countably
infinite) then we adopt the discrete topology, in which case the continuity
requirements in (a) and (b) are satisfied automatically, while the compactness
requirement on $\Gamma$ is satisfied if $\Gamma(x)$ is finite for each $x$.  

Condition (c) imposes the natural requirement that higher continuation values
increase lifetime values, while condition (d) is a consistency requirement
that allows $w_1$ and $w_2$ to act as lower and upper bounds for lifetime
value.  The conditions in assumption~\ref{a:ath} are held to be true
throughout the remainder of the paper. 

Let $\Sigma$ be a family of maps from $\XX$ to $\AA$, referred to below as the
set of all \emph{feasible policies}, such that each $\sigma \in \Sigma$ is
Borel measurable and satisfies $\sigma(x) \in \Gamma(x)$ for all $x \in \XX$.

\begin{lemma}
    \label{l:cst}
    The map $w(x) := H(x, \sigma(x), v)$ is an element of $\vV$ for all $v \in \vV$.   
\end{lemma}

\begin{proof}
    Borel measurability of $(x, a) \mapsto H(x, a, v)$ and $\sigma$ imply that 
    $w$ is Borel measurable on $\XX$.  Moreover, since
    $w_1 \leq v$, the inequalities in \eqref{eq:mon} and \eqref{eq:uls} imply
    $w_1(x) \leq H(x, \sigma(x), w_1) \leq H(x, \sigma(x), v)$ for all $x$.  In particular,
    $w_1 \leq w$.  A similar argument gives $w \leq w_2$, so $w \in \vV$.
\end{proof}

Given $\sigma \in \Sigma$, a function $v_\sigma \in \vV$ that satisfies
\begin{equation}
    \label{eq:vsig}
    v_\sigma(x) = H(x, \sigma(x), v_\sigma) 
    \quad \text{for all } x \in \XX
\end{equation}
is called a \emph{$\sigma$-value function}.  The value $v_\sigma(x)$ can be
interpreted as the lifetime value of following policy $\sigma$.
Its existence and uniqueness are discussed below.

\subsection{Maximization}

\label{ss:max}

We begin by studying maximization of value.  Our key assumption is that the state-action aggregator
satisfies value-convexity and possesses a strict upper solution:

\begin{assumption}
    [Convex Program]
    \label{a:cvx}
    The following conditions are satisfied:
    \begin{enumerate}
        \item $H$ is value-convex.
        \item There exists an $\epsilon > 0$ such that $H(x, a, w_2) \leq w_2(x) - \epsilon$ for all $(x, a) \in \GG$.
    \end{enumerate}
\end{assumption}

Note that part (b) is a strengthening of one of the conditions in
\eqref{eq:uls}.

\begin{proposition}
    \label{p:regp}
    If assumption~\ref{a:cvx} holds, then, for each $\sigma$ in $\Sigma$,
    the set $\vV$ contains exactly one $\sigma$-value function $v_\sigma$.
\end{proposition}

Proposition~\ref{p:regp} assures us that the value $v_\sigma$ of a given
policy $\sigma$ is well defined.  From this foundation we can introduce
optimality concerning a maximization decision problem.  
In particular, in the present setting, a policy $\sigma^* \in \Sigma$ is
called \emph{optimal} if 
\begin{equation*}
    v_{\sigma^*}(x) \geq v_\sigma(x)
    \quad \text{for all } \sigma \in \Sigma
    \text{ and all } x \in \XX.
\end{equation*}
The \emph{maximum value function} associated with this planning problem is the
map $v^*$ defined at $x \in \XX$ by
\begin{equation}
    \label{eq:vstar}
    v^*(x) = \sup_{\sigma \in \Sigma} v_\sigma(x).
\end{equation}
One can show from conditions (c) and (d) of assumption~\ref{a:ath} that $v^*$
is well defined as a real valued function on $\XX$ and satisfies $w_1 \leq v^*
\leq w_2$.

A function $v \in \vV$ is said to satisfy the \emph{Bellman equation} if
\begin{equation}
    \label{eq:belleq}
    v(x) = \max_{a \in \Gamma(x)} H(x, a, v)
    \; \text{ for all } x \in \XX.
\end{equation}
The \emph{Bellman operator} $T$ associated with our abstract dynamic program
is a map sending $v$ in $\cC$ into
\begin{equation}
    \label{eq:bellop} 
    T v(x) = \max_{a \in \Gamma(x)} H(x, a, v).
\end{equation}
Since $v$ is in $\cC$, existence of the maximum is guaranteed by assumption~\ref{a:ath}.  It
follows from Berge's theorem of the maximum that $Tv$ is an element of $\cC$.  
Evidently solutions to the Bellman equation in $\cC$ exactly coincide with
fixed points of $T$.

The convex program conditions lead to the following central result:

\begin{theorem}
    \label{t:bkvx}
    If assumption~\ref{a:cvx} holds, then 
    \begin{enumerate}
        \item The Bellman equation has exactly one solution in $\cC$ and that
            solution is $v^*$.
        \item If $v$ is in $\cC$, then $T^n v \to v^*$ uniformly on $\XX$ as $n \to \infty$.
        \item A policy $\sigma$ in $\Sigma$ is optimal if and only if 
            \begin{equation*}
                \sigma(x) \in \argmax_{a \in \Gamma(x)} H(x, a, v^*)
                \; \text{ for all } x \in \XX.
            \end{equation*}
        \item At least one optimal policy exists.
    \end{enumerate}
\end{theorem}

The fixed point and convergence results for $T$ in theorem~\ref{t:bkvx} rely
on a fixed point theorem for monotone convex operators due to
\cite{du1990}.  There, convergence is shown to be uniformly geometric: 
there exist constants $\lambda \in (0, 1)$ and $K \in \RR$ such
that
\begin{equation*}
    \| T^n v - v^* \| \leq \lambda^n K
    \text{ for all } n \in \NN \text{ and } v \in \cC.
\end{equation*}

\subsection{Minimization}

\label{ss:min}

Next we treat minimization.  In this setting, the convexity and strict upper
solution in assumption~\ref{a:cvx} are replaced by concavity and a strict
lower solution.  
In order to maintain consistency with other sources, we admit some overloading
of terminology relative to section~\ref{ss:max} on maximization. For example,
the optimal policy will now reference a minimizing policy rather than a maximizing
one, and the Bellman equation will shift from maximization to minimization.
The relevant definition will be clear from context.

The next assumption is analogous to assumption~\ref{a:cvx}, which was used for
maximization.

\begin{assumption}
    [Concave Program]
    \label{a:ccv}
    The following conditions are satisfied:
    \begin{enumerate}
        \item $H$ is value-concave.
        \item There exists an $\epsilon > 0$ such that $H(x, a, w_1) \geq w_1(x) + \epsilon$ for all $(x, a) \in \GG$.
    \end{enumerate}
\end{assumption}

Note that part (b) is a strengthening of one of the conditions in
\eqref{eq:uls}.

\begin{proposition}
    \label{p:regp-ccv}
    If assumption~\ref{a:ccv} holds, then, for each $\sigma$ in $\Sigma$,
    the set $\vV$ contains exactly one $\sigma$-value function $v_\sigma$.
\end{proposition}

Proposition~\ref{p:regp-ccv} mimics proposition~\ref{p:regp}, assuring us
that, in the present context, the cost $v_\sigma$ of a given policy $\sigma$ is well defined.  
A policy $\sigma^* \in \Sigma$ is then called \emph{optimal} if 
\begin{equation*}
    v_{\sigma^*}(x) \leq v_\sigma(x)
    \quad \text{for all } \sigma \in \Sigma
    \text{ and all } x \in \XX.
\end{equation*}
The \emph{minimum cost function} associated with this
planning problem is the function $v^*$
defined at $x \in \XX$ by
\begin{equation}
    \label{eq:vstar-min}
    v^*(x) = \inf_{\sigma \in \Sigma} v_\sigma(x).
\end{equation}

A function $v \in \vV$ is said to satisfy the \emph{Bellman equation} if
\begin{equation}
    \label{eq:belleq-min}
    v(x) = \min_{a \in \Gamma(x)} H(x, a, v)
    \qquad \text{for all } x \in \XX. 
\end{equation}
The \emph{Bellman operator} $S$ associated with our abstract dynamic program
is a map sending $v$ in $\cC$ into
\begin{equation}
    \label{eq:bellop-min} 
    S v(x) = \min_{a \in \Gamma(x)} H(x, a, v) .
\end{equation}
Analogous to theorem~\ref{t:bkvx}, we have

\begin{theorem}
    \label{t:bkcv}
    If assumption~\ref{a:ccv} holds, then 
    \begin{enumerate}
        \item The Bellman equation \eqref{eq:belleq-min} has exactly one
            solution in $\cC$ and that solution is the minimum cost function $v^*$.
        \item If $v$ is in $\cC$, then $S^n v \to v^*$ uniformly on $\XX$ as $n \to \infty$.
        \item A policy $\sigma$ in $\Sigma$ is optimal if and only if 
            \begin{equation*}
                \sigma(x) \in \argmin_{a \in \Gamma(x)} H(x, a, v^*)
                \; \text{ for all } x \in \XX.
            \end{equation*}
        \item At least one optimal policy exists.
    \end{enumerate}
\end{theorem}

\section{Applications}

\label{s:a}

In this section we study a collection of applications, showing how the general
results in section~\ref{s:grs} can be used to solve the dynamic programming
problems discussed in the introduction.

\subsection{An Additively Separable Decision Process}

\label{ss:addsep-app}

It is worth
noting that the results stated above can be applied in the standard additive separable
case, alongside the traditional Bellman--Blackwell contraction mapping
approach to dynamic programming.  To see this, consider the generic 
dynamic programming model of \cite{stokey1989}, with Bellman equation 
\begin{equation}
    \label{eq:belleq-app-adse}
    v (s,z) 
    = \max_{y \in \Gamma(s, z)} \left\{ F(s,y,z) 
    + \beta \int v (y,z') P(z, \diff z') \right\}
\end{equation}
over $(s,z) \in \SS \times \ZZ$.
Here $\SS$ and $\ZZ$ are compact metric spaces containing possible values for the 
endogenous and exogenous state variables, respectively. 
Let the transition function $P$ on $\ZZ$ have the Feller property, 
let the feasible correspondence $\Gamma \colon \SS \times \ZZ \to \SS$ be 
compact valued and continuous, let $F \colon \GG \to \RR$ be continuous,
and let $\beta$ lie in $(0,1)$.

We translate this model to our environment by taking $x := (s, z)$ to be
the state, $\XX := \SS \times \ZZ$ to be the state space, $a = y \in \SS$ to be the
action, and setting 
\begin{equation*}
    H((s, z), y, v) = F(s,y,z) + \beta \int v (y, z') P(z, \diff z'). 
\end{equation*}
Since $F$ is continuous on a compact set, there exists a finite constant $M$
with $|F| \leq M$.\footnote{The domain $\GG$ of $F$ is compact in the product
topology by Tychonoff's theorem.}  For the bracketing functions $w_1$ and
$w_2$ we fix $\epsilon > 0$ and adopt the constant functions 
\begin{equation*}
    w_1 \equiv - \frac{M}{1-\beta}
    \quad \text{and} \quad
    w_2 \equiv \frac{M + \epsilon}{1-\beta}.   
\end{equation*}
The conditions of assumption~\ref{a:ath} are all satisfied.  Conditions (a)
and (b) are true by assumption and condition (c) is trivial to verify. 
To see that condition (d) of assumption~\ref{a:ath} holds, 
we note that $w_1$ and $w_2$ lie in $bc\XX$. 
In addition, for any given $((s, z), y) \in \GG$, we have
\begin{equation*}
    H((s,z), y,w_1) 
    = F(s,y,z) - \beta \frac{M}{1-\beta} 
    \geq - M - \beta \frac{M}{1-\beta} 
    = w_1(s, z).
\end{equation*}
Similarly,
\begin{equation*}
    H((s, z), y, w_2)
    = F(s,y,z) + \beta \frac{M+\epsilon}{1-\beta}
    \leq M + \beta \frac{M+\epsilon}{1-\beta} 
    = w_2(s, z) -\epsilon.
\end{equation*}
The last inequality gives not only $H((s,z), y, w_2) \leq w_2(s, z)$, as required for
part (d) of assumption, but also the stronger condition in 
part (b) of assumption~\ref{a:cvx}.  Thus, to verify the requirements of
theorem~\ref{t:bkvx}, we need only check the convexity condition in part (a)
of assumption~\ref{a:cvx}.  But this is immediate from the linearity of
expectations.  Hence theorem~\ref{t:bkvx} applies.

\subsection{Epstein-Zin Preferences}

\label{ss:ezp-app}

\cite{epstein1989} propose a
specification of lifetime value that separates and independently parameterizes
intertemporal elasticity of substitution and risk aversion.  
Value is defined recursively by the CES aggregator
\begin{equation*}
    U_t
    = \left[ (1-\beta)C^{1-\rho}_t + \beta \left\{ \rR_t (U_{t+1}) \right\}^{1-\rho}
      \right]^{\frac{1}{1-\rho}}
      \quad (0 < \rho \neq 1),
\end{equation*}
where $\{ C_t \}$ is a consumption path, $U_t$ is the utility value of the path onward from
time $t$, and $\rR_t$ is the Kreps-Porteus certainty equivalent operator
\begin{equation*}
    \rR_t(U_{t+1})
    = \left( \EE_t U^{1-\gamma}_{t+1} \right)^{\frac{1}{1-\gamma}}
    \quad (0 < \gamma \neq 1).
\end{equation*}
Here, $\EE_t$ stands for the conditional expectation with respect to the period
$t$ information.
The value $1/ \rho$ represents elasticity of intertemporal substitution (EIS)
between the composite good and the certainty equivalent, while $\gamma$
governs the level of relative risk aversion (RRA) with respect to atemporal
gambles.  The most empirically relevant case is 
$\rho < \gamma$, implying that the agent prefers early
resolution of uncertainty (see, e.g., \cite{bansal2004yaron} or
\cite{schorfheide2018identifying}).  We focus on this case in what follows.

Under Epstein--Zin preferences, the generic additively separable
Bellman equation in \eqref{eq:belleq-app-adse} becomes
\begin{equation}
    \label{eq:belleq-app-ez}
    v(s, z) 
    = \max_{y \in \Gamma(s,z)} \left\{
        r(s, y, z)
      + \beta \left[ \int v(y, z')^{1-\gamma} P(z, \diff z') 
      \right]^{\frac{1-\rho}{1-\gamma}} \right\}^{\frac{1}{1-\rho}}
\end{equation}
for each $(s,z) \in \SS \times \ZZ$, where, here and below,
\begin{equation*}
    r(s, y, z) := (1-\beta) F(s,y,z)^{1-\rho}.
\end{equation*}
We impose the same conditions on the primitives discussed in
section~\ref{ss:addsep-app}.  In particular,  $F$ is continuous, $P$ is
Feller, $\Gamma$ is continuous and compact valued and both $\SS$ and $\ZZ$ are
compact.  To ensure $F(s, y, z)^{1-\rho}$ is always
well defined, we also assume that $F$ is everywhere positive.

\subsubsection{The Case $\rho < \gamma < 1$}

\label{sss:theta01}

As in \cite{hansen2012recursive}, we begin with the continuous strictly
increasing transformation $\hat v = v^{1-\gamma}$, which allows us to
rewrite~\eqref{eq:belleq-app-ez} as
\begin{equation}
    \label{eq:v-hat-max}
    \hat v(s,z)  
    = \max_{y \in \Gamma(s,z)} \left\{
        r(s, y, z)
      + \beta  \left[ \int \hat v(y,z') P(z, \diff z')
      \right]^{1/\theta} \right\}^\theta
\end{equation}
where 
\begin{equation*}
    \theta:= \frac{1-\gamma}{1-\rho}.     
\end{equation*}
Since this transformation is bijective, there is a one-to-one
correspondence between $v$ and $\hat v$, in the sense that $v$ solves
\eqref{eq:belleq-app-ez} if and only if $\hat v$ solves \eqref{eq:v-hat-max}.
Note that in the current setting we have $\theta \in (0,1)$.

The state-action aggregator $H$ corresponding to \eqref{eq:v-hat-max} is
\begin{equation}
    \label{eq:q-ez}
    H((s,z), y, v) 
    = \left\{ r(s, y, z)
      + \beta  \left[ \int v(y,z') P(z, \diff z')
      \right]^{1/\theta} \right\}^\theta.
\end{equation}
For the bracketing functions $w_1$ and $w_2$, we fix $\delta > 0$ and take the constant functions
\begin{equation*}
    w_1  
    := \left( \frac{m}{1-\beta} \right)^\theta
    \quad \text{and} \quad
    w_2 := \left( \frac{M + \delta}{1-\beta} \right)^\theta ,
\end{equation*}
where 
\begin{equation}
    \label{eq:mm}
    m :=  \min_{((s,z),y) \in \GG}  r(s, y, z) 
    \; \text{ and } \;
    M :=  \max_{((s,z),y) \in \GG} r(s, y, z) .
\end{equation}
These values are finite and positive, since $F$ is continuous
and positive on a compact domain.\footnote{
    In this case, positivity of $F$ can be weakened to nonnegativity.
    }  
Being constant, $w_1$ and $w_2$ are continuous.

We now show that the conditions of assumptions~\ref{a:ath} and \ref{a:cvx} are all satisfied.
Regarding assumption~\ref{a:ath}, condition (a) is true by assumption, while
condition (b) follows immediately from the continuity imposed on $F$
and the Feller property of $P$.
Condition (c) is easy to verify, since, for any $b > 0$, the scalar map
\begin{equation}
    \label{eq:psi}
    \psi(t) := (b + \beta t^{1/\theta})^\theta
    \qquad (t \geq 0)
\end{equation}
is monotone increasing.  
To check condition (d), observe that, for fixed $((s,z),y) \in \GG$, we have
\begin{equation*}
    H((s,z),y,w_1) 
    = \left\{ 
        r(s, y, z) + \beta  \frac{m}{1-\beta} 
      \right\}^\theta 
    \geq \left\{ m + \beta \frac{m}{1-\beta}  \right\}^\theta = w_1(s,z).
\end{equation*}
Similarly,
\begin{equation*}
    H((s,z),y,w_2) 
    = \left\{ 
        r(s, y, z)
      + \beta  \frac{M + \delta}{1-\beta}  \right\}^\theta 
    \leq \left\{ M + \beta \frac{M+\delta}{1-\beta} \right\}^\theta,
\end{equation*}
or, with some rearranging,
\begin{equation}
    \label{eq:suw2-ez}
    H((s,z),y,w_2) \leq \left\{ \frac{M+\delta}{1-\beta} - \delta \right\}^\theta
    < w_2(s,z).
\end{equation}
Hence condition (d) of assumption~\ref{a:ath} holds.
In fact, \eqref{eq:suw2-ez} implies that our choice of $w_2$ also
satisfies the uniformly strict inequality in (b) of assumption~\ref{a:cvx}.\footnote{To be precise, condition (b) holds when $\epsilon := [(M+\delta)/(1-\beta)]^\theta - [(M+\delta)/(1-\beta) -
  \delta]^\theta$.}

It only remains to check value-convexity of $H$.  But this
is implied by the convexity of $\psi$ defined in
\eqref{eq:psi}, which holds whenever $0 < \theta \leq 1$, along with linearity of the integral.
The conclusions of theorem~\ref{t:bkvx} now follow.

\subsubsection{The Case $\rho < 1 < \gamma$}

\label{sss:theta<0}

To treat this case we again apply the continuous transformation $\hat v \equiv v^{1-\gamma}$ to
the Bellman equation~\eqref{eq:belleq-app-ez}.  But now $1 - \gamma$ is
negative, leading to the \emph{minimization} problem
\begin{equation}
    \label{eq:v-hat-min}
    \hat v(s,z)  
    = \min_{y \in \Gamma(s,z)} \left\{
        r(s, y, z)
      + \beta  \left[ \int \hat v(y,z') P(z, \diff z')
      \right]^{1/\theta} \right\}^\theta
\end{equation}
for each $(s, z) \in \XX$.  
The state-action aggregator $H$ corresponding to~\eqref{eq:v-hat-min}
is still as defined in~\eqref{eq:q-ez}.
Note that in the current setting we have $\theta < 0$.

As \eqref{eq:v-hat-min} is a minimization problem, we aim to apply theorem~\ref{t:bkcv}.
For the bracketing functions $w_1$ and $w_2$, 
we take the constant functions
\begin{equation*}
    w_1  
    := \left( \frac{M + \delta}{1-\beta} \right)^\theta 
    \quad \text{and} \quad
    w_2 
    := \left( \frac{m}{1-\beta} \right)^\theta ,
\end{equation*}
where $\delta$ is a positive constant and $m$ and $M$ are as defined in 
\eqref{eq:mm}.

The conditions of assumptions~\ref{a:ath} and \ref{a:ccv} are all satisfied.
Regarding assumption~\ref{a:ath}, 
the arguments verifying conditions (a) to (c)
are identical to those in section~\ref{sss:theta01}.
To check condition (d), observe that, for fixed $((s,z),y) \in \GG$, we have 
\begin{equation*}
    H((s,z),y,w_1) 
    = \left\{ r(s, y, z)
      + \beta  \frac{M+\delta}{1-\beta} \right\}^\theta 
    \geq \left\{ M + \beta \frac{M + \delta}{1-\beta} \right\}^\theta,
\end{equation*}
or, with some rearranging,
\begin{equation}
    \label{eq:slw1-ez1}
    H((s,z),y,w_1)
    \geq \left\{ \frac{M+\delta}{1-\beta} - \delta \right\}^\theta
    > w_1(s,z).
\end{equation}
Similarly, for fixed $((s,z),y) \in \GG$, we have
\begin{equation*}
    H((s,z),y,w_2) 
    = \left\{ r(s, y, z)
      + \beta  \frac{m}{1-\beta} \right\}^\theta 
    \leq \left\{ m + \beta \frac{m}{1-\beta} \right\}^\theta,
\end{equation*}
and the last term is equal to $w_2(s,z)$.
Hence, condition (d) of assumption~\ref{a:ath} is verified.
Furthermore, it is immediate from~\eqref{eq:slw1-ez1} that
our choice of $w_1$ also satisfies the uniformly strict inequality
in (b) of assumption~\ref{a:ccv}.


It only remains to check the value-concavity of $H$.
But this follows directly from the concavity of the function $\psi$
defined in~\eqref{eq:psi}, as implied by  $\theta < 0$, along with linearity of the integral.
We have now checked all conditions of theorem~\ref{t:bkcv}.

\subsubsection{The Case $1 < \rho < \gamma$}

\label{sss:theta>1}

We now turn to the model in the case where the coefficient of relative risk aversion
is still strictly great than $1$ but intertemporal elasticity of substitution is
less than $1$, as is commonly found in the literature.\footnote{See, for
    example, \cite{farhi2008werning} or
\cite{basu2017bundick}.} As before, we apply the continuous transformation $\hat v \equiv v^{1-\gamma}$
to the Bellman equation~\eqref{eq:belleq-app-ez} and, since $1-\gamma < 0$,
the transformed counterpart leads us to the minimization problem 
as defined in~\eqref{eq:v-hat-min}.  Note that $\theta > 1$ in the current setting.

As \eqref{eq:v-hat-min} is a minimization problem, 
we aim to apply theorem~\ref{t:bkcv}. 
For the bracketing functions $w_1$ and $w_2$, 
we take the constant functions
\begin{equation*}
    w_1  
    := \left( \frac{m - \delta}{1-\beta} \right)^\theta 
    \quad \text{and} \quad
    w_2 
    := \left( \frac{M}{1-\beta} \right)^\theta,
\end{equation*}
for some positive $\delta < m$, where $m$ and $M$ are as defined in \eqref{eq:mm}.

Assumptions~\ref{a:ath} and \ref{a:ccv} are again satisfied.  Regarding assumption~\ref{a:ath}, 
the arguments of verifying conditions (a) to (c) 
are identical to those in section~\ref{sss:theta01}.
To check condition (d), observe that,
for fixed $((s,z),y) \in \GG$, we have
\begin{equation*}
    H((s,z),y,w_1) 
    = \left\{ r(s, y, z)
      + \beta  \frac{m - \delta}{1-\beta} \right\}^\theta 
    \geq \left\{ m + \beta \frac{m- \delta}{1-\beta}  \right\}^\theta,
\end{equation*}
or, with some rearranging
\begin{equation}
    \label{eq:slw1-ez2}
    H((s,z),y,w_1)
    \geq \left\{ \frac{m - \delta}{1-\beta} + \delta \right\}^\theta
    > w_1 (s,z).
\end{equation}
Similarly, for fixed $((s,z),y) \in \GG$, we have
\begin{equation*}
    H((s,z),y,w_2) 
    = \left\{ r(s, y, z)
      + \beta \frac{M}{1-\beta} \right\}^\theta 
    \leq \left\{ M + \beta \frac{M}{1-\beta} \right\}^\theta
     = w_2(s,z).
\end{equation*}
Hence condition (d) of assumption~\ref{a:ath} holds.
In fact \eqref{eq:slw1-ez2} implies that
our choice of $w_1$ also satisfies the uniformly strict inequality 
in (b) of assumption~\ref{a:ccv}.

Value-concavity of $H$ is a direct consequence of the concavity of 
$\psi$, which holds again when $\theta > 1$, along
with linearity of the integral.  The conclusions of theorem~\ref{t:bkcv}
now follow.

\subsection{Risk Sensitive Preferences}

\label{ss:riskssp-app}

Consider an agent with risk sensitive preferences (see, e.g.,
\cite{bauerle2018stochastic}), leading to Bellman equation 
\begin{equation}
    \label{eq:belleq-rsp}
    v(s,z) 
    = \max_{y \in \Gamma(s,z)} \left\{
      r(s,y,z) 
      - \frac{\beta}{\theta} \ln \left[ \int \exp \left(-\theta \, v(y,z') \right)
      P(z, \diff z') \right] \right\}
\end{equation}
for each $(s,z) \in \SS \times \ZZ$.
Here, $r \colon \GG \to \RR$ is a continuous one-period reward function.
The parameter $\theta > 0$ captures the risk sensitivity,
while other primitives are as discussed in section~\ref{ss:addsep-app}.
In particular, $P$ is Feller, 
$\Gamma$ is continuous and compact valued and both $\SS$ and $\ZZ$ are compact.


Applying the continuous bijective transformation
$\hat v \equiv \exp(-\theta v)$ to $v$ in the Bellman equation~\eqref{eq:belleq-rsp}
leads to the minimization problem
\begin{equation}
    \label{eq:v-hat-min-rsp}
    \hat v(s,z) 
    = \min_{y \in \Gamma(s,z)} \exp \left\{ - \theta \left\{
      r(s,y,z) 
      - \frac{\beta}{\theta} \ln \left[ \int \hat v(y,z')
      P(z, \diff z') \right] \right\} \right\}.
\end{equation}
We translate \eqref{eq:v-hat-min-rsp} to our environment
by taking $\XX := \SS \times \ZZ$ to be the state space,
$a=y \in \SS$ to be the action, and setting 
\begin{equation}
    \label{eq:q-rsp}
    H((s,z), y, v) 
    = \exp 
    \left\{ 
        - \theta \left\{
      r(s,y,z) 
      - \frac{\beta}{\theta} \ln \left[ \int v(y,z')
      P(z, \diff z') \right] \right\} 
    \right\}.
\end{equation}
Since $r$ is continuous, there exists a finite constant $M$ with $| r| \leq M$.
For the bracketing functions, we fix $\delta > 0$ and take the constant functions 
\begin{equation*}
    w_1 
    := \exp \left[ - \theta \left( \frac{M}{1-\beta} +\delta \right) \right] 
    \quad \text{and} \quad
    w_2 
    := \exp \left[ - \theta \left( \frac{-M}{1-\beta} \right) \right].
\end{equation*}

Assumptions~\ref{a:ath} and \ref{a:ccv} are all satisfied.
Regarding assumption~\ref{a:ath}, the steps
verifying conditions (a) and (b) are identical to those
in section~\ref{sss:theta01}.
Condition (c) clearly holds, since, for any $b \in \RR$, the scalar map
\begin{equation}
    \label{eq:psi-rsp}
    \phi(t)
    := \exp \left[ - \theta \left( b - \frac{\beta}{\theta} \ln t \right) \right]
    \qquad ( t > 0 )
\end{equation}
is monotone increasing.
To check condition (d), observe that,
for fixed $((s,z),y) \in \GG$, we have
\begin{align*}
    H((s,z),y,w_1) 
    &= \exp \left( - \theta \left\{
      r(s,y,z) + \beta \left( \frac{M}{1-\beta} +\delta \right)
       \right\} \right) \\
    &\geq \exp \left( - \theta \left\{
      M + \beta \left( \frac{M}{1-\beta} +\delta \right)
       \right\} \right)
\end{align*}
or, with some rearranging,
\begin{equation}
    \label{eq:slw1-rsp}
    H((s,z),y,w_1)
    \geq \exp \left( -\theta \left\{ 
      \frac{M}{1-\beta} + \beta \delta \right\} \right) 
    > w_1(s,z).
\end{equation}
Similarly, for fixed $((s,z),y) \in \GG$, we have
\begin{align*}
    H((s,z),y,w_2) 
    & = \exp \left( - \theta \left\{
      r(s,y,z) 
      + \beta \left( \frac{-M}{1-\beta} \right) \right\} \right) 
      \\
    & \leq \exp \left( - \theta \left\{
      - M - \beta \frac{M}{1-\beta} \right\} \right),
\end{align*}
and the last term is equal to $w_2(s,z)$.
Hence condition (d) of assumption~\ref{a:ath} holds.
In addition, it is obvious from~\eqref{eq:slw1-rsp} that
our choice of $w_1$ also satisfies the uniformly strict inequality 
in part (b) of assumption~\ref{a:ccv}.\footnote{
    To be precise, condition (b) of assumption~\ref{a:ccv} holds 
    when we set $\epsilon := \exp \{ -\theta [ M/(1-\beta) + \beta \delta ] \}
    - \exp \{ -\theta [ M/(1-\beta) +  \delta ] \}$.
    }

Finally, condition (a) of assumption~\ref{a:ccv}, which is 
value-concavity of $H$, follows from directly the concavity of the function
$\phi$ defined in~\eqref{eq:psi-rsp}, along with linearity of the
integral.
The conclusions of theorem~\ref{t:bkcv} now follow.

\subsection{Ambiguity}
    
\label{ss:sap-app}

\cite{ju2012ambiguity} propose and study a recursive smooth ambiguity 
model where lifetime value satisfies
\begin{equation}
    \label{eq:EIS-def-sap}
    U_t
    = \left[ (1-\beta)C_t^{1-\rho} + \beta \{ \rR_t (U_{t+1}) \}^{1-\rho} \right]^{1/(1-\rho)}
\end{equation}
with
\begin{equation}
    \label{eq:AR-def-sap}
    \rR_t (U_{t+1} )
    = \left\{  \EE_{\mu_t} \left( \EE_{\pi_{\theta, t}} \left[U^{1-\gamma}_{t+1}
    \right] \right)^{(1-\eta)/(1-\gamma)}    
    \right\}^{1/(1-\eta)}.
\end{equation}
As before, $\rho$ is the reciprocal of the EIS and
$\gamma$ governs risk aversion, 
while $\eta$ satisfies $0< \eta \neq 1$ and captures ambiguity aversion.
If $\eta = \gamma$, the decision maker is ambiguity neutral and
~\eqref{eq:EIS-def-sap}--\eqref{eq:AR-def-sap} reduces to the classical
recursive utility model of \cite{epstein1989}.  The decision maker displays
ambiguity aversion if and only if $\gamma < \eta$.
We focus primarily on the case $0< \rho \leq 1 < \gamma < \eta$, which is the
most empirically relevant.\footnote{The calibration used in \cite{ju2012ambiguity}
is $(\rho, \gamma, \eta) = (0.66, 2.0, 8.86)$.  See p.~574.}

As a generic formulation of the preferences of
\cite{ju2012ambiguity}, we consider the Bellman equation
\begin{equation}
    \label{eq:bellm-ori-sap}
    v(s, z) 
    = \max_{y \in \Gamma(s, z)}
    \left\{ r(s, y, z) 
        + \beta \left\{ \int 
        \left[  \int v(y, z')^{1-\gamma} \pi_\theta (z, \diff z')
        \right]^{\frac{1-\eta}{1-\gamma}} \mu(z, \diff \theta)
        \right\}^{\frac{1-\rho}{1-\eta}}
    \right\}^{\frac{1}{1-\rho}}
\end{equation}
where $(s, z) \in \SS \times \ZZ$.  
We assume both $\SS$ and $\ZZ$ to be compact,
$\Gamma$ to be continuous and compact valued,
$F$ to be continuous and everywhere positive.  The set
$\Theta$ is a finite parameter space, each element of which 
is a vector of parameters in the specification of the exogenous state process.
Given any $\theta \in \Theta$, the transition function $\pi_\theta$ on $\ZZ$
is assumed to have the Feller property.  Given any $z \in \ZZ$,
the distribution $\mu(z, \cdot)$ maps subsets of $\Theta$
to $[0,1]$ and evolves as a function of the exogenous state process.  
We suppose that $\mu$ is continuous in $z$ for each $\theta \in \Theta$.

\subsubsection{The Case $\rho \neq 1$.}

\label{sss:rho<1-sap}

Applying the continuous bijective transformation $\hat v \equiv v^{1-\eta}$ to $v$
in the Bellman equation~\eqref{eq:bellm-ori-sap} leads to the minimization problem
\begin{equation}
    \label{eq:bellm-modif-sap}
    \hat v(s, z) 
    = \min_{y \in \Gamma(s, z)}
    \left\{ r(s, y, z) 
    + \beta \left\{ \int 
    \left[  \int \hat v(y, z')^{\xi_1} \pi_\theta (z, \diff z')
    \right]^{\frac{1}{\xi_1}} \mu(z, \diff \theta)
    \right\}^{\frac{1}{\xi_2}}
    \right\}^{\xi_2}
\end{equation}
for all $(s, z) \in \SS \times \ZZ$,
where, here and below,
\begin{equation*}
    \xi_1 : = \frac{1-\gamma}{1-\eta}
    \quad \text{and} \quad 
    \xi_2 := \frac{1-\eta}{1-\rho}.
\end{equation*}
Since this transformation is bijective,
there is a one-to-one correspondence between $v$ and $\hat v$,
in the sense that $v$ solves \eqref{eq:bellm-ori-sap}
if and only if $\hat v$ solves \eqref{eq:bellm-modif-sap}.
Note that in the current setting we have $\xi_1 \in (0,1)$
and $\xi_2 < 0$.

We translate this model to our environment
by taking $\XX := \SS \times \ZZ$ to be the state space, 
$a=y$ to be the action taking values in $\SS$,
and setting the state-action aggregator $H$ to
\begin{equation}
    \label{eq:q-def-sap}
    H((s,z),y, \hat v)
    = \left\{ r(s, y, z) 
    + \beta \left\{ \int 
    \left[  \int \hat v(y, z')^{\xi_1} \pi_\theta (z, \diff z')
    \right]^{\frac{1}{\xi_1}} \mu(z, \diff \theta)
    \right\}^{\frac{1}{\xi_2}}
    \right\}^{\xi_2}.
\end{equation}
As \eqref{eq:bellm-modif-sap} is a minimization problem,
we aim to apply theorem~\ref{t:bkcv}. 
For the bracketing functions $w_1$ and $w_2$,
we fix $\delta >0$ and take the constant functions
\begin{equation*}
    w_1
    := \left( \frac{M +\delta}{1-\beta} \right)^{\xi_2} 
    \quad \text{and} \quad
    w_2
    := \left( \frac{m}{1-\beta} \right)^{\xi_2} ,
\end{equation*}
where the real numbers $m$ and $M$ are as defined in section~\ref{sss:theta01}.

Assumptions~\ref{a:ath} and \ref{a:ccv} are satisfied.  Regarding
assumption~\ref{a:ath}, condition (a) is true by assumption.  Conditions (b)
and (c) are proved in lemma~\ref{l:q-sap} in the appendix.  To verify
condition (d), for fixed $((s,z),y) \in \GG$, we have
\begin{align*}
    H((s,z),y,w_1) 
    &= \left\{ r(s,y,z)
      + \beta \left\{ \int 
       \left( \frac{M +\delta}{1-\beta} \right)^{\xi_2} 
      \mu (z, \diff \theta)
      \right\}^{1/\xi_2}
      \right\}^{\xi_2} \\
    &= \left\{ r(s,y,z) + \beta \frac{M+\delta}{1-\beta} \right\}^{\xi_2} 
    \geq \left\{ M + \beta \frac{M+\delta}{1-\beta} \right\}^{\xi_2},
\end{align*}
where the first equality follows from directly the definition of $H$
and the fact that $[\int d(z')^{\xi_1} \pi_\theta(z, \diff z')]^{1/\xi_1} =d$ for 
any non-negative constant function $d$.
Furthermore, with some rearranging, we obtain
\begin{equation}
    \label{eq:slw1-sap}
    H((s,z),y,w_1) 
    \geq \left\{ \frac{M+\delta}{1-\beta} - \delta \right\}^{\xi_2} 
    > w_1(s,z).
\end{equation}
Similarly, for fixed $((s,z),y) \in \GG$, we have
\begin{align*}
    H((s,z),y,w_2) 
    &= \left\{ r(s,y,z)
      + \beta \left\{ \int 
       \left( \frac{m}{1-\beta} \right)^{\xi_2} 
      \mu (z, \diff \theta)
      \right\}^{1/\xi_2}
      \right\}^{\xi_2} \\
    &= \left\{ r(s,y,z) + \beta \frac{m}{1-\beta}  \right\}^{\xi_2}
    \leq \left\{ m + \beta \frac{m}{1-\beta}  \right\}^{\xi_2} = w_2(s,z).
\end{align*}
Hence condition (d) of assumption~\ref{a:ath} indeed holds true.
Moreover, it is clear from~\eqref{eq:slw1-sap} that our choice
of $w_1$ also satisfies the uniformly strict inequality in part (b)
of assumption~\ref{a:ccv}.

Condition (a) of assumption~\ref{a:ccv} (i.e., value-concavity of $H$) is also
satisfied, as shown in lemma~~\ref{l:q-sap} of the appendix.  
The conclusions of theorem~\ref{t:bkcv} now follow.

\subsubsection{The Case $\rho =1$.}

\label{sss:rho=1-sap}

In the limiting case with $\rho = 1$, the generic ambiguity recursion
\eqref{eq:AR-def-sap} becomes
\begin{equation}
    \label{eq:rsu-sap-rho=1}
    \begin{split}
    U_t(C)
    = & (1-\beta) \ln C_t \\
      &+  \frac{\beta}{1-\eta} \ln \left\{ 
        \EE_{\mu_t} \exp \left( \frac{1-\eta}{1-\gamma} \ln \left( \EE_{\pi_{\theta, t}}
        \exp\left( (1-\gamma) U_{t+1} \right) \right) \right) \right\},
    \end{split}
\end{equation}
where $U_t = \ln V_t$.\footnote{This specification connects with
risk-sensitive control and robustness, as studied by \cite{hansen2008robustness}.
In particular, there are two risk-sensitivity adjustments 
in~\eqref{eq:rsu-sap-rho=1}.} The generic Bellman equation in~\eqref{eq:bellm-ori-sap}
becomes
\begin{equation}
    \label{eq:bell-sap=1}
    \begin{split}
    &v(s,z)
     = \max_{y \in \Gamma(s,z)}  \bigg\{ r(s,y,z) 
        +  \frac{\beta}{1-\eta} \times \\
             &\times    \ln \left[  \int \exp \left(
                 \frac{1}{\xi_1}  
                 \ln \left( 
                 \int \exp \left( (1-\gamma) v(y,z') \right) \pi_\theta (z, \diff z')
                 \right) \right) 
                 \mu(z, \diff \theta)  \right] \bigg\},
    \end{split}
\end{equation}
for each $(s,z) \in \SS \times \ZZ$. 
The one-period return function $r$ is still assumed to be continuous but no longer restricted to be positive,
while other primitives are as discussed in
section~\ref{sss:rho<1-sap}.

Applying the transformation $\hat v \equiv \exp[ (1-\eta) v]$ to $v$
in the Bellman equation~\eqref{eq:bell-sap=1} leads us to the 
minimization problem
\begin{equation} 
    \label{eq:v-hat-min-sap=1}
    \begin{split}
    &\hat v(s,z)
    = \min_{y \in \Gamma(s,z)}
      \exp \Biggl( (1-\eta)   \bigg\{ r(s,y,z)  
      + \frac{\beta}{1-\eta} \times \\
            &\times     \ln \left[ \int \exp \left(
                 \frac{1}{\xi_1}  
                 \ln \left( 
                 \int \exp \left( \xi_1 \ln \hat v(y,z') 
                 \right) \pi_\theta (z, \diff z')
                 \right) \right) 
                 \mu(z, \diff \theta)  \right] \bigg\}  \Biggr).
    \end{split}
\end{equation}
With some rearranging, \eqref{eq:v-hat-min-sap=1} can be written as
\begin{equation} 
    \label{eq:v-hat-min-sap=1-altversion}
    \begin{split}
    \hat v(s,z)
    = \min_{y \in \Gamma(s,z)}
      \exp & \Bigg( (1-\eta)  \Bigg\{  r(s,y,z)
    + \frac{\beta}{1-\eta} \times \\
     &\times \ln \left[ \int 
      \left( \int \hat v (y,z')^{\xi_1} \pi_\theta (z, \diff z') \right)^{1/\xi_1}
      \mu(z, \diff \theta)
      \right] \Bigg\} \Bigg).
    \end{split}
\end{equation}
Note that we still have $\xi_1 \in (0,1)$ and $\eta > 1$
in the current setting with ambiguity aversion.
The state-action aggregator $H$ corresponding to~\eqref{eq:v-hat-min-sap=1-altversion} is
\begin{equation} 
    \label{eq:q-sap-rho=1}
    \begin{split}
    H((s,z), y, \hat v)
    = 
    \exp & \Bigg( (1-\eta)  \Bigg\{  r(s,y,z) 
    + \frac{\beta}{1-\eta} \times \\
     &\times \ln \left[ \int 
      \left( \int \hat v (y,z')^{\xi_1} \pi_\theta (z, \diff z') \right)^{1/\xi_1}
      \mu(z, \diff \theta)
      \right] \Bigg\} \Bigg).
    \end{split}
\end{equation}

Since the return function $r$ is continuous on a compact set, there exists
a finite constant $M$ such that $|r| \leq M$.
Hence for the bracketing function $w_1$ and $w_2$, 
we fix $\delta > 0$ and take the constant functions
\begin{equation*}
    w_1 
    := \exp \left( (1-\eta) \left(  \frac{M}{1-\beta} + \delta \right) \right)
    \quad \text{and} \quad
    w_2 
    := \exp \left( (1-\eta) \left(  \frac{-M}{1-\beta} \right) \right).
\end{equation*}

As \eqref{eq:v-hat-min-sap=1-altversion} is the minimization problem,
we aim to apply theorem~\ref{t:bkcv}.
Again, assumptions~\ref{a:ath} and \ref{a:ccv} are all satisfied.

Regarding assumption~\ref{a:ath}, condition (a) is trivial.
Conditions (b) and (c) follow from lemma~\ref{l:q-sap-rho=1} in the
appendix.  To check condition (d), observe that,
for fixed $((s,z), y) \in \GG$, we have
\begin{align*}
    H((s,z), y, w_1) 
    &= \exp \left( (1-\eta) \left\{ r(s,y,z)
      + \beta \left(  \frac{M}{1-\beta} + \delta \right)
     \right\} \right) \\
    & \geq \exp \left( (1-\eta) \left\{ M
      + \beta \left(  \frac{M}{1-\beta} + \delta \right)
     \right\} \right),
\end{align*}
or, with some rearranging,
\begin{equation}
    \label{eq:slw1-sap=1}
    H((s,z), y, w_1) 
    \geq \exp \left( (1-\eta) \left\{  \frac{M}{1-\beta} + \beta \delta \right\} \right)
    > w_1 (s,z).
\end{equation}
Similarly, for fixed $((s,z), y) \in \GG$, we have
\begin{align*}
    H((s,z),y,w_2) 
    &= \exp \left( (1-\eta) \left\{
      r(s,y,z) 
      + \beta \left( \frac{-M}{1-\beta} \right) \right\} \right) \\
    &\leq \exp \left( (1-\eta) \left\{
      - M - \beta \frac{M}{1-\beta} \right\} \right) = w_2(s,z).
\end{align*}
Hence condition (d) of assumption~\ref{a:ath} holds.

In fact, it is immediate from~\eqref{eq:slw1-sap=1} that
our choice of $w_1$ also satisfies the uniformly strict inequality 
in part (b) of assumption~\ref{a:ccv}.
Regarding part (a) of assumption~\ref{a:ccv}, 
value-concavity of $H$ is immediate from 
lemma~\ref{l:q-sap-rho=1}.
We have now checked all conditions of theorem~\ref{t:bkcv}, and hence
the conclusions of that theorem now follow.

\subsection{Narrow framing}

\label{ss:narfram}

In this section we study recursive preferences that incorporate both
first-order risk aversion and narrow framing,
as in, say, \cite{barberis2006AER} or \cite{barberis2009huang}.
They can be expressed as
\begin{equation*}
    U_t
    = \left[ (1-\beta)C^{1-\rho}_t + \beta \left\{ 
     \left( \EE_t \, U^{1-\gamma}_{t+1} \right)^{\frac{1}{1-\gamma}} 
     + b_0 \EE_t \left( \sum_{i=1} \bar u( \tilde G_{i,t+1}) \right) \right\}^{1-\rho}
      \right]^{\frac{1}{1-\rho}} ,
\end{equation*}
where $b_0 \geq 0$ is a parameter controlling the degree of narrow framing, 
while $\tilde G_{i,t+1}$ represents the specific gamble the agent is taking by investing in asset $i$
whose uncertainty will be resolved between period $t$ and $t+1$.
First-order risk aversion is introduced 
through the piecewise linearity of $\bar u$.\footnote{
    The specification of $\bar u$ 
    in~\cite{barberis2006AER} is 
    $\bar u (x) =  x \1\{x \geq 0\} + \lambda x \1\{x < 0\}$ with $\lambda >
    1$.} 
Relative to the recursive specification in section~\ref{ss:ezp-app},
the new term prefixed by $b_0$ shows that the agent obtains utility directly
from the outcomes of gambles $\{ \tilde G_{i,t+1}\}_i$ over and above what those outcomes mean
for total wealth risk, rather than just indirectly via its contribution to next period's wealth.
Other primitives are as discussed in section~\ref{ss:ezp-app}.
For the parameters $\rho$ and $\gamma$, we assume that either $1 < \rho < \gamma$
or $\rho < 1 < \gamma$.

Under the preceding preference specification, the generic Bellman equation becomes
\begin{equation}
    \label{eq:belleq-app-ez-nf}
    v(s, z) 
    = \max_{y \in \Gamma(s,z)} \left\{
        r(s, y, z)
      + \beta \left[ 
      \left( \int v(y, z')^{1-\gamma} P(z, \diff z') \right)^{\frac{1}{1-\gamma}}
      + B(s,y,z)
      \right]^{1-\rho} \right\}^{\frac{1}{1-\rho}}.
\end{equation}
As before, we suppose that 
the one-period return function $r(s,y,z)$ is positive and continuous on $\GG$,
while the aggregate gambling utility function $B(s,y,z)$
is assumed to be nonnegative and continuous on $\GG$.

We again apply the continuous transformation $\hat v \equiv v^{1-\gamma}$ to
the Bellman equation~\eqref{eq:belleq-app-ez-nf}.  As $1 - \gamma$ is
negative, the transformed counterpart leads us to the minimization problem
\begin{equation}
    \label{eq:vhat-app-ez-nf}
    \hat v(s, z) 
    = \min_{y \in \Gamma(s,z)} \left\{
        r(s, y, z)
      + \beta \left[ 
      \left( \int \hat v(y, z') P(z, \diff z') \right)^{\frac{1}{1-\gamma}}
      + B(s,y,z)
      \right]^{1-\rho} \right\}^\theta, 
\end{equation}
where $\theta := (1-\gamma)/(1-\rho)$.
The state-action aggregator is
\begin{equation}
    \label{eq:q-app-ez-nf}
    H((s,z), y, \hat v)
    = \left\{
        r(s, y, z)
      + \beta \left[ 
      \left( \int \hat v(y, z') P(z, \diff z') \right)^{\frac{1}{1-\gamma}}
      + B(s,y,z)
      \right]^{1-\rho} \right\}^\theta. 
\end{equation}

\begin{lemma}
    \label{l:slu-nf}
    Let $H$ be as defined in~\eqref{eq:q-app-ez-nf}.
    If either $\rho < 1 < \gamma$ or $1 < \rho < \gamma$,
    then there exist continuous strictly positive functions 
    $w_1, w_2$ on $\SS \times \ZZ$, $w_1 < w_2$, such that 
    \begin{enumerate}
        \item[(SL)] there exists an $\epsilon >0$ such that
            $H((s,z),y,w_1) \geq w_1(s,z) + \epsilon$ for all $((s,z),y) \in \GG$; and
        \item[(U)] $H((s,z),y,w_2) \leq w_2(s,z)$ for all $((s,z),y) \in \GG$.
    \end{enumerate}
\end{lemma}

The proof is deferred to the appendix.

The conditions of assumptions~\ref{a:ath} and \ref{a:ccv} are all satisfied.
Regarding assumption~\ref{a:ath}, condition (a) is true by assumption, while condition (b)
follows immediately from the continuity imposed on $r$ and $B$ and the Feller property of $P$.
Condition (c) is easy to verify, since, for any fixed constants $c>0$ and $b \geq 0$,
the scalar map
\begin{equation}
    \label{eq:psi-nf}
    \psi(t) 
    := \left\{ c + \beta \left[ t^{\frac{1}{1-\gamma}} + b \right]^{1-\rho}
        \right\}^\theta
        \qquad (t>0)
\end{equation}
is monotone increasing.  Condition (d) and part (b) of assumption~\ref{a:ccv}
have been verified by lemma~\ref{l:slu-nf}.

It only remains to check value-concavity of $H$.  But this follows directly from the concavity
of $\psi$ defined in~\eqref{eq:psi-nf}, as implied by $\rho < 1 < \gamma$,
along with linearity of the integral.
Hence all conditions of theorem~\ref{t:bkcv} are verified and the conclusions of
that theorem follow.

\section{Unbounded Rewards}

\label{s:ubdd}

This section gives an example of how the dynamic programming methodology
proposed in this paper can be extended to the
setting of unbounded rewards.  The application we consider is the Epstein-Zin
problem of section~\ref{sss:theta>1}, where $1 < \rho < \gamma$,
with Bellman equation 
\begin{equation*}
    v(s, z) 
    = \max_{y \in \Gamma(s,z)} \left\{
      r(s,y,z)
      + \beta \left[ \int v(y, z')^{1-\gamma} P(z, \diff z') 
      \right]^{\frac{1-\rho}{1-\gamma}} \right\}^{\frac{1}{1-\rho}}
\end{equation*}
for each $(s,z) \in \SS \times \ZZ$.
Dropping the compactness assumption, we allow 
$\SS$ and $\ZZ$ to be arbitrary separable metric spaces containing possible
values for the endogenous and exogenous state variables respectively.  As
before, $P$ is Feller, $\Gamma$ is compact valued and continuous, while $r \colon \GG \to
\RR$ is continuous and $\beta$ lies in $(0,1)$.  Let $\theta$ be defined by
$\theta = (1-\gamma)/(1-\rho)$, so that $\theta > 1$ in the current setting.

In addition, we make the following assumptions.
\begin{assumption}
    \label{a:app-ext}
    There exist a continuous function $\kappa \colon \SS \times \ZZ \to [1, \infty)$,
    positive constants $M$, $L$ with $L \leq M$,
    and $c \in (0, 1/\beta^\theta)$ and $d \in [0, 1/\beta^\theta)$
    satisfying the conditions
    \begin{equation}
        \label{eq:a-ext-sf}
        \sup_{y \in \Gamma(s,z)} r(s,y,z) \leq M \kappa (s,z)
        \quad \text{for all } (s,z) \in \SS \times \ZZ,
    \end{equation}
    \begin{equation}
        \label{eq:a-ext-if}
        \inf_{y \in \Gamma(s,z)} r(s,y,z) \geq L \kappa (s,z)
        \quad \text{for all } (s,z) \in \SS \times \ZZ,
    \end{equation}
    \begin{equation}
        \label{eq:a-ext-sk}
        \sup_{y \in \Gamma(s,z)} \int \kappa(y, z')^\theta P(z, \diff z') 
        \leq c \kappa (s,z)^\theta
        \quad \text{for all } (s,z) \in \SS \times \ZZ,
    \end{equation}
    \begin{equation}
        \label{eq:a-ext-ik}
        \inf_{y \in \Gamma(s,z)} \int \kappa(y, z') P(z, \diff z') 
        \geq d \kappa (s,z)
        \quad \text{for all } (s,z) \in \SS \times \ZZ.
    \end{equation}
    Moreover, the map
    $(y, z) \mapsto \int \kappa(y,z')^\theta P(z, \diff z')$ is continuous on $\SS \times \ZZ$.
\end{assumption}

In what follows, given a real-valued continuous function $\ell$ defined on $\SS \times \ZZ$
with $\ell (s,z) > 0$ for all $(s,z) \in \SS \times \ZZ$, a function $f \colon
\SS \times \ZZ \to \RR$ is called \emph{$\ell$-bounded} if $f(s,z) / \ell
(s,z)$ is bounded as $(s,z)$ ranges over $\SS \times \ZZ$.


As in section~\ref{sss:theta>1}, 
we apply the continuous transformation $\hat v \equiv v^{1-\gamma}$
to the Bellman equation.  Since $1-\gamma < 0$,
the transformed counterpart leads us to a minimization problem 
with state-action aggregator 
\begin{equation*}
    H((s,z), y, \hat v) 
    = \left\{ r(s,y,z)
      + \beta  \left[ \int \hat v(y,z') P(z, \diff z')
      \right]^{1/\theta} \right\}^\theta.
\end{equation*}
For the bracketing functions $w_1$ and $w_2$,
we fix $\delta$ such that $0<\delta< L$, and then adopt the functions
\begin{equation*}
    w_1 := (L - \delta )^\theta \cdot \kappa
        \quad \text{and} \quad
    w_2 := \left( \frac{M}{1-\beta c^{1/\theta}} 
        \right)^\theta \cdot \kappa^\theta.
\end{equation*}
Note that $\kappa \leq \kappa^\theta$, since $\theta >1$ and $\kappa \geq 1$.
Hence $w_1$ and $w_2$ are both $\kappa^\theta$-bounded. 
In addition, the positivity and the continuity of $\kappa$ directly imply
the positivity and continuity of $w_1$ and $w_2$. 
Hence, such $w_1$ and $w_2$ are 
positive $\kappa^\theta$-bounded continuous functions in $\RR^\XX$ 
with $w_1 \leq w_2$.

Let $\vV$ and $\cC$ be all Borel measurable functions $v$ in $\RR^\XX$ satisfying
$w_1 \leq v \leq w_2$, and be the continuous functions in $\vV$ respectively.
For fixed $\sigma \in \Sigma$, a function $\hat v_\sigma \in \vV$ is called
$\sigma$-value function if
\begin{equation*}
    \hat v_\sigma(s,z)
    = \left\{ r(s,\sigma(s,z),z)
      + \beta  \left[ \int \hat v_\sigma(\sigma(s,z),z') P(z, \diff z')
      \right]^{1/\theta} \right\}^\theta
\end{equation*} 
for all $(s,z) \in \SS \times \ZZ$.  The following proposition states a result for
its existence and uniqueness.

\begin{proposition} 
    \label{p:regp-ccv-ext}
    If assumption~\ref{a:app-ext} holds, then, for each $\sigma \in \Sigma$,
    the set $\vV$ contains exactly one $\sigma$-value function $\hat v_\sigma$.
\end{proposition}

From this foundation, the minimum cost function $\hat v^*$ associated with this planning
problem is defined at $(s,z) \in \SS \times \ZZ$ by 
    $\hat v^*(s,z) 
    = \inf_{\sigma \in \Sigma} \hat v_\sigma (s,z)$.
A function $\hat v \in \vV$ is said to satisfy the Bellman equation if
\begin{equation*}
    \hat v(s,z) 
    = \min_{y \in \Gamma(s,z)}
      \left\{ r(s,y,z)
      + \beta  \left[ \int \hat v(y,z') P(z, \diff z')
      \right]^{1/\theta} \right\}^\theta.
\end{equation*}
In this connection, the corresponding Bellman operator $S$ on $\cC$ is defined through 
\begin{equation*}
    S \hat v(s,z) 
    = \inf_{y \in \Gamma(s,z)}
      \left\{ r(s,y,z)
      + \beta  \left[ \int \hat v(y,z') P(z, \diff z')
      \right]^{1/\theta} \right\}^\theta.
\end{equation*}

Analogous to the result in bounded case (see, e.g., section~\ref{sss:theta>1}), we have

\begin{theorem}
    \label{t:bkcv-ext}
    If assumptions~\ref{a:app-ext} holds, then 
    \begin{enumerate}
        \item The minimum cost function $\hat v^*$ lies in $\cC$ 
            and is the unique solution of the Bellman equation~\eqref{eq:v-hat-min} in that set.
        \item If $\hat v$ is in $\cC$, then $S^n \hat v \to \hat v^*$ as $n \to \infty$.
        \item A policy $\sigma$ in $\Sigma$ is optimal if and only if 
            \begin{equation*}
                \sigma(s,z) \in \argmin_{y \in \Gamma(s,z)} H((s,z), y, \hat v^*)
                \: \text{ for all } (s,z) \in \SS \times \ZZ.
            \end{equation*}
        \item At least one optimal policy exists.
    \end{enumerate}
\end{theorem}





\section{Appendix}

Let $m\XX$ represent all Borel measurable functions in $\RR^\XX$ and let
$c\XX$ denote all continuous functions in $m\XX$.  Let $bm\XX$ be the bounded
functions in $m\XX$ and let $bc\XX$ be the continuous functions in $bm\XX$.
Let $m\XX_+$ and $m\XX_{++}$ be the non-negative and positive functions in
$m\XX$, respectively.
Recall that a self map $A$ on a convex subset $M$ of $bm\XX$ is called
\begin{itemize}
    \item \emph{asymptotically stable} on $M$ if $A$ has a unique fixed point
        $v^*$ in $M$ and $A^n v \to v^*$ as $n \to \infty$ whenever $v \in M$,
    \item \emph{isotone} if $A v \leq A v'$ whenever $v, v' \in M$ with $v \leq v'$, 
    \item \emph{convex} if $A (\lambda v + (1-\lambda) v') \leq \lambda A v +
        (1-\lambda) A v'$ whenever $v, v' \in M$ and $0 \leq \lambda \leq 1$, and
    \item \emph{concave} if $A (\lambda v + (1-\lambda) v') \geq \lambda A v +
        (1-\lambda) A v'$ whenever $v, v' \in M$ and $0 \leq \lambda \leq 1$.
\end{itemize}

For $f, g \in bm\XX$, the statement $f \ll g$ means that there exists an
$\epsilon > 0$ such that $f(x) \leq g(x) - \epsilon$ for all $x \in \XX$.

For each $\sigma \in \Sigma$, we define the \emph{$\sigma$-value operator} $T_\sigma$ on $\vV$ by
\begin{equation}
    \label{eq:Tsig}
    T_\sigma v (x) := H(x, \sigma(x), v)
    \quad \text{for all } x \in \XX , \, v \in \vV.
\end{equation}
Stating that $v_\sigma \in \vV$ solves~\eqref{eq:vsig} is equivalent to
stating that $v_\sigma$ is a fixed point of $T_\sigma$.  By lemma~\ref{l:cst},
the operator $T_\sigma$ is a well defined self-map on $\vV$.

\subsection{Proofs for the Convex Case}

\begin{lemma}
    \label{l:sts}
    If assumption~\ref{a:cvx} holds, then, for each $\sigma \in \Sigma$, the
    operator $T_\sigma$ is asymptotically stable on $\vV$.
\end{lemma}

\begin{proof}[Proof of lemma~\ref{l:sts}]
    Fix $\sigma \in \Sigma$.  We aim to apply theorem 3.1 of \cite{du1990}.  To this end, 
    it is sufficient to show that 
    \begin{enumerate}
        \item[(i)] $T_\sigma$ is isotone and convex on $\vV$.
        \item[(ii)] $T_\sigma w_1 \geq w_1$ and $T_\sigma w_2 \ll w_2$.
    \end{enumerate}
    Regarding condition (i), pick any $v, v' \in \vV$ with $v \leq v'$. For
    fixed $x \in \XX$,
    we have 
    \begin{equation*}
        T_\sigma v(x) = H(x, \sigma(x), v) 
         \leq H(x, \sigma(x), v') = T_\sigma v'(x),
    \end{equation*}
    by \eqref{eq:mon}.  Hence, isotonicity of $T_\sigma$ holds.

    To see that $T_\sigma$ is convex, fix $v, v' \in \vV$ and $\lambda \in [0,1]$.
    For any given $x \in \XX$, we have
    \begin{align*}
        T_\sigma (\lambda v + (1-\lambda) v') (x)
         &= H(x, \sigma(x), \lambda v + (1-\lambda) v') \\
         &\leq \lambda H(x, \sigma(x), v) + (1-\lambda) H(x, \sigma(x), v') \\
         &= \lambda T_\sigma v(x) + (1-\lambda) T_\sigma v'(x),
    \end{align*}
    where the inequality directly follows from part (a) of assumption~\ref{a:cvx}.
    Since $x \in \XX$ was arbitrary, the convexity of $T_\sigma$ follows.

    The first part of condition (ii)  follows directly from~\eqref{eq:uls}, since, for each $x \in \XX$,
    \begin{equation*}
        T_\sigma w_1 (x) = H(x, \sigma(x), w_1) \geq w_1(x).
    \end{equation*}
    To see that the second part of condition (ii) is satisfied, it follows from part (b) of assumption~\ref{a:cvx} that
    \begin{equation*}
        T_\sigma w_2 (x) = H(x, \sigma(x), w_2)
         \leq w_2(x) - \epsilon
    \end{equation*}
    for each $x \in \XX$ and for some $\epsilon >0$.  Hence $w_2 \gg T_\sigma w_2$, as was to be shown.
\end{proof}

\begin{proof}[Proof of proposition~\ref{p:regp}]
    This follows directly from lemma~\ref{l:sts}.
\end{proof}

Given $v \in \vV$, a policy $\sigma$ in $\Sigma$ will be called
\emph{$v$-maximal-greedy} if
\begin{equation}
    \label{eq:dvg}
    \sigma(x) \in \argmax_{a \in \Gamma(x)} H(x, a, v)
    \; \text{ for all } x \in \XX.
\end{equation}

\begin{lemma}
    \label{l:egp}
    If $v \in \cC$, then there exists at least one $v$-maximal-greedy policy.
\end{lemma}

\begin{proof}
    Fixing $v \in \cC$ and using the compactness and continuity conditions in
    assumption~\ref{a:ath}, we can choose for each $x \in \XX$ an action
    $\sigma(x) \in \Gamma(x)$ such that \eqref{eq:dvg} holds.
    The map $\sigma$ constructed in this manner
    can be chosen to be Borel measurable by 
    theorem 18.19 of \cite{aliprantis2006border}.
\end{proof}

\begin{lemma}
    \label{l:st}
    If assumption~\ref{a:cvx} holds, then $T$ is asymptotically stable on
    $\cC$.
\end{lemma}

\begin{proof}[Proof of lemma~\ref{l:st}]
    In order to apply theorem 3.1 of \cite{du1990},
    it suffices to show that
    \begin{enumerate}
        \item[(i)] $T$ is isotone and convex on $\cC$.
        \item[(ii)] $Tw_1 \geq w_1$ and $Tw_2 \ll w_2$.
    \end{enumerate}
    The isotonicity of $T$ on $\cC$ is obvious, since,
    by the monotonicity restriction~\eqref{eq:mon},
    \begin{align*}
        v \leq v'
        &\implies \max_{a \in \Gamma(x)} H(x, a, v)
        \leq \max_{a \in \Gamma(x)} H(x, a, v)
        \quad \text{for all } x \in \XX.
    \end{align*}
    In other words, by definition of $T$,
    $v \leq v'$ implies $Tv \leq Tv'$.

    To show the convexity of $T$, 
    fix $v, v' \in \cC$ and $\lambda \in [0,1]$.
    For any given $(x, a) \in \GG$, 
    we have, by part (a) of assumption~\ref{a:cvx},
    \begin{align*}
        H(x,a, \lambda v + (1-\lambda) v')
        &\leq \lambda H(x,a,v) + (1-\lambda) H(x,a,v') \\
        &\leq \lambda \max_{a \in \Gamma(x)} H(x,a,v) 
          + (1-\lambda) \max_{a \in \Gamma(x)}H(x,a,v') \\
        &= \lambda Tv(x) + (1-\lambda) Tv'(x).
    \end{align*}
    Since $(x,a) \in \GG$ was arbitrary, 
    the above inequality implies
    \begin{equation*}
        \max_{a \in \Gamma(x)} H(x,a, \lambda v + (1-\lambda) v')
        \leq \lambda Tv(x) + (1-\lambda) Tv'(x)
    \end{equation*}
    for each $x \in \XX$, which in turn means that
    $T [ \lambda v + (1-\lambda) v' ] \leq \lambda Tv + (1-\lambda) Tv'$.

    The first part of condition (ii)
    follows directly from~\eqref{eq:uls}, since, for each $x \in \XX$,
    \begin{equation*}
        Tw_1(x) 
        = \max_{a \in \Gamma(x)} H(x, a, w_1)
        \geq H(x, a, w_1)
        \geq w_1 (x).
    \end{equation*}

    To see that the second part of condition (ii) is satisfied,
    it follows from part (b) of assumption~\ref{a:cvx} that
    \begin{equation*}
        T w_2(x) 
        = \max_{a \in \Gamma(x)} H(x,a,w_2) \leq w_2 (x) - \epsilon
    \end{equation*}
    for each $x \in \XX$ and for some $\epsilon >0$.
    Hence, $T w_2 \ll w_2$, as was to be shown.
\end{proof}

\begin{theorem}
    \label{t:ccv-app}
    If $T_\sigma$ is asymptotically stable on $\vV$ for all $\sigma \in \Sigma$ and
    $T$ is asymptotically stable on $\cC$, then the conclusions of
    theorem~\ref{t:bkvx} hold.
\end{theorem}

\begin{proof}
    Let $v^*$ be the maximum value function and let $\bar v$ be the unique fixed point
    of $T$ in $\cC$.  To see that $\bar v = v^*$, first observe that $\bar v \in \cC$ and hence $\bar v$ has at
    least one maximal-greedy policy $\sigma$.  For this policy we have, by definition,
    $T_\sigma \bar v(x) = T \bar v(x)$ at each $x$, from which it follows that
    $\bar v = T \bar v = T_\sigma \bar v$.  Since $T_\sigma$ is asymptotically
    stable on $\vV$, we know that its unique fixed point is $v_\sigma$, so
    $\bar v = v_\sigma$.  But then $\bar v \leq v^*$, by the definition of
    $v^*$.

    To see that the reverse inequality holds, pick any $\sigma \in \Sigma$.  We have 
    $T_\sigma \bar v \leq T \bar v = \bar v$.  Iterating on this inequality
    and using the isotonicity of $T_\sigma$ gives $T^k_\sigma \bar v \leq \bar
    v$ for all $k$.  Taking the limit with respect to $k$ and using the
    asymptotic stability of $T_\sigma$ then gives $v_\sigma \leq \bar v$.  
    Hence $v^* \leq \bar v$, and we can now conclude that $\bar v = v^*$.

    Since $\bar v \in \cC$, we have $v^* \in \cC$.    It follows that $v^*$ is the unique solution
    to the Bellman maximization equation in $\cC$, 
    and that $T^n v \to v^*$ whenever $v \in \cC$.  
    Parts (a) and (b) of theorem~\ref{t:bkvx} are now established. 

    Regarding part (c) and (d), by the definition of maximal-greedy policies, we know
    that $\sigma$ is $v^*$-maximal-greedy iff $H(x, \sigma(x), v^*) = \max_{a \in
    \Gamma(x)} H(x, a, v^*)$ for all $x \in \XX$.  Since $v^*$ satisfies the
    Bellman maximization equation, we then have
    \begin{equation*}
        \sigma \text{ is $v^*$-maximal-greedy}
        \quad \iff \quad 
        H(x, \sigma(x), v^*) = v^*(x),
            \;\; \forall \, x \in \XX.
    \end{equation*}
    But, by proposition~\ref{p:regp}, the right hand side is equivalent to the
    statement that $v^* = v_\sigma$.  Hence, by this chain of logic and the
    definition of optimality,
    \begin{equation}
        \sigma \text{ is $v^*$-maximal-greedy}
        \iff v^* = v_\sigma
        \iff \text{ $\sigma$ is optimal}
    \end{equation}
    Moreover, the fact that $v^*$ is in $\cC$ combined with lemma~\ref{l:egp} assures us that at least one $v^*$-maximal-greedy
    policy exists.  Each such policy is optimal, 
    so the set of optimal policies is nonempty.
\end{proof}

\subsection{Proofs for the Concave Case}

We prove the minimization results from the maximization results.
To begin, recall the definitions from the beginning of section~\ref{s:grs}
and let $\check w_1 := - w_2$ and $\check w_2 := - w_1$.
Clearly, $\check w_1$ and $\check w_2$ are bounded continuous functions in $\RR^\XX$ 
satisfying $\check w_1 \leq \check w_2$.
We denote by $\check \vV := - \vV$ 
the set of all Borel measurable functions $\check v$ in $\RR^\XX$ 
satisfying $\check w_1 \leq \check v \leq \check w_2$,
and let $\check \cC := -\cC$ be the continuous functions in $\check \vV$.
The \emph{conjugate} state-action aggregator 
$\check H \colon \GG \times \check \vV \to \RR$
is defined by
\begin{equation}
    \label{eq:checkH-def}
    \check H (x,a, \check v) := - H (x,a, -\check v).
\end{equation}

\begin{lemma}
    \label{l:checkH}

    If the state-action aggregator $H$ satisfies
    assumptions~\ref{a:ath} and \ref{a:ccv},
    then the conjugate aggregator $\check H$ defined in~\eqref{eq:checkH-def}
    satisfies assumptions~\ref{a:ath} and \ref{a:cvx}.
\end{lemma}

\begin{proof}[Proof of lemma~\ref{l:checkH}]
 It is clear that
 conditions (a) and (b) in assumption~\ref{a:ath} hold true for $\check H$.
 Regarding condition (c), pick any $\check v, \check v'$ in $\check \vV$
 with $\check v \leq \check v'$.
 Observe that $- \check v $ and $- \check v'$ are in $\vV$ 
 satisfying $- \check v' \leq - \check v$.
 By condition (c) of assumption~\ref{a:ath} for the aggregator $H$, we have
 \begin{align*}
    - \check v' \leq - \check v
    &\implies
    H(x,a, -\check v') \leq H(x,a, -\check v) \\
    &\iff
    -H(x,a, -\check v') \geq -H(x,a, -\check v) \\
    &\iff
    \check H(x,a, \check v') \geq \check H(x,a, \check v)
    \quad \text{for all} \quad
    (x,a) \in \GG,
 \end{align*}
 which implies that $\check H$ indeed satisfies monotonicity condition~\eqref{eq:mon}.

 Regarding the upper-lower bounds condition (d) for $\check H$, 
 invoking condition~\eqref{eq:uls} for $H$ and after some rearrangement,
 we then have $\check H(x,a, -w_1) \leq -w_1(x)$
 and $\check H(x,a, -w_2) \geq -w_2(x)$ for all $(x,a) \in \GG$.
 By the definitions of $\check w_1$ and $\check w_2$,
 it follows that $\check w_1(x) \leq \check H(x,a, \check w_1)$
 and $\check H(x,a, \check w_2) \leq \check w_2(x)$
 for all $(x,a) \in \GG$, as desired.

 It remains to show that the aggregator $\check H$ satisfies value-convexity
 and possesses a strict upper solution.
 To show the value-convexity of $\check H$, 
 fix $\check v, \check v' \in \check \vV$ and $\lambda \in [0,1]$.
 Clearly, $-\check v$ and $-\check v'$ are in $\vV$.
 Then, for any given $(x,a) \in \GG$, by the value-concavity of $H$,
 we have 
 \begin{align*}
    H(x,a, \lambda(-\check v) + (1-\lambda)(-\check v') )
    \geq \lambda H(x,a, -\check v) + (1-\lambda) H(x,a, -\check v'),
 \end{align*}
 or 
 \begin{align*}
    - H(x,a, \lambda(-\check v) + (1-\lambda)(-\check v') )
    \leq - \lambda H(x,a, -\check v) - (1-\lambda) H(x,a, -\check v'),
 \end{align*}
 It follows immediately from the definition of $\check H$ that
 $\check H(x,a, \lambda \check v + (1-\lambda) \check v') 
 \leq \lambda \check H(x,a, \check v) + (1-\lambda) \check H(x,a, \check v')$,
 as was to be shown.
 
 To verify that $\check w_2$ is a strict upper solution for $\check H$,
 we invoke part (b) of assumption~\ref{a:ccv} for $H$.
 That is, there exists an $\epsilon > 0$ such that 
 $H(x,a, w_1) \geq w_1(x) + \epsilon$ for all $(x,a) \in \GG$.
 By the definition of $\check H$, it is equivalent to stating that  
 there exists an $\epsilon > 0$ such that
 \begin{align*}
    - \check H(x,a, -w_1) \geq w_1(x) + \epsilon 
    \iff 
    \check H(x,a, -w_1) \leq - w_1(x) - \epsilon
    \quad \text{for all} \quad (x,a) \in \GG.
 \end{align*}
 Invoking the definition of $\check w_2$, 
 it is equivalent to saying that
 there exists an $\epsilon > 0$ such that
 $ \check H(x,a, \check w_2) \leq \check w_2(x) - \epsilon$ for all $(x,a) \in \GG$,
 which completes the proof of lemma~\ref{l:checkH}.
\end{proof}

In addition, we observe that the $\sigma$-value function equation~\eqref{eq:vsig}
can be expressed equivalently as
$\check v_\sigma(x) = \check H(x, \sigma(x), \check v_\sigma)$
for all $x \in \XX$
with $\check v_\sigma := - v_\sigma \in \check \vV$.
Furthermore, the minimum cost function equation~\eqref{eq:vstar-min} can
be expressed equivalently as
    $\check v^* (x) = \sup_{\sigma \in \Sigma} \check v_\sigma (x)$
with $\check v^* := - v^*$.
The Bellman equation~\eqref{eq:belleq-min} can be 
expressed equivalently as
    $\check v(x) 
    = \max_{a \in \Gamma(x)} \check H(x,a, \check v)$
with $\check v := - v \in \check \vV$.
Hence, by virtue of lemma~\ref{l:checkH},
we can apply theorem~\ref{t:bkvx} to the conjugate aggregator $\check H$
and recover the results stated in theorem~\ref{t:bkcv} for the aggregator $H$.

\subsection{Proofs for Section~\ref{ss:sap-app}}

\label{ss:prf-app-apdix}

In this section, we prove some properties of the state-action aggregator $H$
defined in section~\ref{ss:sap-app}.

For the sake of exposition, 
fix $\theta \in \Theta$,
we first define an operator $R_\theta$ on $bm(\SS \times \ZZ)_+$ by
\begin{equation}
    \label{eq:Rthe-def}
    (R_\theta w) (y, z)
    := \left[  \int w(y, z')^{\xi_1} \pi_\theta (z, \diff z')
    \right]^{1/\xi_1}
    \quad \text{for all } (y, z) \in \SS \times \ZZ.
\end{equation}
From this foundation, we then define an operator $R$ that is a map
sending $(\theta, w)$ in $\Theta \times bm(\SS \times \ZZ)_+$ into
\begin{equation}
    \label{eq:R-def}
    R w (y,z, \theta) := (R_\theta w) (y, z)
    \quad \text{for all } (y,z, \theta) \in \SS \times \ZZ \times \Theta.
\end{equation}

The following lemma shows some useful properties of the operator $R_\theta$.
\begin{lemma}
    \label{l:r-theta}
    For fixed $\theta \in \Theta$, if $\xi_1$ lies in $(0,1)$,
    then the operator $R_\theta$ defined in \eqref{eq:Rthe-def}
    is isotone and concave on $bm(\SS \times \ZZ)_+$. 

    Moreover, the function $R_\theta w$ is nonnegative, bounded, and
    Borel measurable on $\SS \times \ZZ$ whenever $w \in bm(\SS \times \ZZ)_+$
    and continuous on $\SS \times \ZZ$
    whenever $w \in bc(\SS \times \ZZ)_+$.
\end{lemma}

\begin{proof}[Proof of lemma~\ref{l:r-theta}]
    Fix $\theta \in \Theta$.
    The isotonicity of $R_\theta$ is obvious, 
    since the scalar function $\RR_+ \ni t \mapsto t^{\xi_1} \in \RR_+$ 
    and its inverse are both strictly increasing on $\RR_+$.

    Since $\xi_1 \in (0,1)$, by Theorem 198 of \cite{hardy1934inequalities},
    we know that $R_\theta$ is super-additive in the sense that
    for any $w, w' \in m(\SS \times \ZZ)_+$,
    $R_\theta (w + w') \geq R_\theta(w) + R_\theta(w')$.\footnote{
        This result can also be reviewed as 
        the reverse Minkowski inequality,
        see, for example, Proposition 3.2 in page 225 of \cite{dibenedetto2002realanalysis}.
        }
    As a result, the super-additivity and the positive homogeneity of $R_\theta$
    together yield the concavity of $R_\theta$.\footnote{
        An operator $A$ defined on the positive cone $bm\XX_+$ of $bm\XX$
        is called \emph{positively homogeneous} (of the first degree) if
        for any $v$ in $bm\XX_+$ and any real number $t \geq 0$,
        we have $A(tv) = t Av$.
        }
    Indeed, pick any $\lambda \in [0,1]$ and $w, w' \in m(\SS \times \ZZ)_+$,
    by the convexity of $m(\SS \times \ZZ)_+$, we have
    \begin{align*}
        R_\theta [ \lambda w + (1-\lambda) w']
        & \geq R_\theta(\lambda w) + R_\theta ((1-\lambda) w')
        \quad (\text{by super-additivity}) \\
        & = \lambda R_\theta (w) + (1-\lambda) R_\theta (w')
        \quad (\text{by positive homogeneity}) ,
    \end{align*}
    as was to be shown.

    Regarding the second claim of lemma~\ref{l:r-theta}, 
    non-negativity and boundedness of $R_\theta w$ is immediate and
    it is easy to see that $R_\theta w$ is Borel measurable on $\SS \times \ZZ$
    whenever $w \in bm(\SS \times \ZZ)_+$.
    Now fix $w \in bc(\SS \times \ZZ)_+$.
    We note that the function $ w^{\xi_1}$ also lies in $bc(\SS \times \ZZ)_+$.
    Then, by virtue of Feller property of $\pi_\theta$,
    the mapping 
    $\SS \times \ZZ \ni (y,z) \mapsto \int w (y, z')^{\xi_1} 
    \pi_\theta(z, \diff z') \in \RR_+$ is bounded and continuous on $\SS \times \ZZ$.
    Furthermore, it follows that the mapping  
    $\SS \times \ZZ \ni (y,z) \mapsto 
    [\int w (y, z')^{\xi_1} \pi_\theta(z, \diff z')]^{1/\xi_1} \in \RR_+$ 
    is continuous on $\SS \times \ZZ$, 
    since the inverse of the map $t \mapsto t^{\xi_1}$
    is also continuous on $\RR_+$.
    Therefore, our claim follows.
\end{proof}

As an application of lemma~\ref{l:r-theta}, we now present the next result.

\begin{lemma}
    \label{l:R-mer}
    The operator $R$ defined in~\eqref{eq:R-def} is a well-defined map
    from $\Theta \times bm(\SS \times \ZZ)_+$ into $bm(\SS \times \ZZ \times \Theta)_+$.
\end{lemma}

\begin{proof}[Proof of lemma~\ref{l:R-mer}]
    Fix $(\theta, w)$ in $\Theta \times bm(\SS \times \ZZ)_+$.
    Since boundedness and non-negativity of the function $Rw$ are obvious,
    it remains to show that $Rw$ is measurable on $\SS \times \ZZ \times \Theta$.

    On one hand, for each $\theta \in \Theta$, it follows from lemma~\ref{l:r-theta} that
    the function $R w ( \cdot, \cdot, \theta) = R_\theta w \colon \SS \times \ZZ \to \RR_+$
    is Borel measurable.
    On the other hand, for each $(y,z) \in \SS \times \ZZ$,
    the function $R w (y, z, \cdot) \colon \Theta \to \RR_+$ is continuous, 
    since $\Theta$ is a finite set (endowed with the discrete topology).

    In this connection, we conclude that 
    the function $Rw \colon \SS \times \ZZ \times \Theta \to \RR$ is a Carath{\'e}odory function,
    in the sense that 
    \begin{enumerate}
        \item[(1)] for each $\theta \in \Theta$, 
            the function $Rw(\cdot, \cdot, \theta) \colon \SS \times \ZZ \to \RR$
            is Borel measurable; and
        \item[(2)] for each $(y,z) \in \SS \times \ZZ$, 
            the function $Rw(y,z, \cdot) \colon \Theta \to \RR$ is continuous.
    \end{enumerate}
    By lemma 4.51 in \cite{aliprantis2006border}, 
    it follows that the Carath{\'e}odory function $Rw$ is jointly measurable on
    $\SS \times \ZZ \times \Theta$, as desired.
\end{proof}

In this connection, the state-action aggregator $H$ defined in~\eqref{eq:q-def-sap} 
can be simply expressed as a composition of two operators $R$ and $\tilde H$
as follows
\begin{equation}
    \label{eq:H-def}
    H((s,z), y , \hat v) 
    := \tilde H((s,z), y, R \hat v),
\end{equation}
with 
\begin{equation}
    \label{eq:tH}
    \tilde H((s,z),y,h)
    := \left\{ r(s, y, z) 
    + \beta \left\{ \int h (y, z, \theta) \mu(z, \diff \theta)
     \right\}^{1/\xi_2}
     \right\}^{\xi_2}
\end{equation}
for all $((s,z),y) \in \GG$ and $h \in bm(\SS \times \ZZ \times \Theta)_{++}$.

It is worth noting that the formula of $\tilde H$ defined in~\eqref{eq:tH}
is almost identical to that of $H$ defined in~\eqref{eq:q-ez}.
Hence, 
recalling the results associated with $H$ in section~\ref{sss:theta<0}, 
we have
\begin{lemma}
    \label{l:tq}
    If $\xi_2 < 0$,
    then $\tilde H$ defined in \eqref{eq:tH} is isotone and concave in its third argument
    on $bm(\SS \times \ZZ \times \Theta)_{++}$.

    In addition, the map $((s,z), y) \mapsto \tilde H ((s,z),y,h)$ is Borel measurable on $\GG$
    whenever $h \in bm(\SS \times \ZZ \times \Theta)_{++}$,
    and continuous on $\GG$
    whenever the map $h(\cdot, \cdot, \theta) \colon \SS \times \ZZ \to \RR_{++}$ is continuous,
    for each $\theta \in \Theta$.
\end{lemma}

\begin{proof}[Proof of lemma~\ref{l:tq}]
    Analogous to the proofs in section~\ref{sss:theta<0},
    for any fixed $b > 0$, we consider the scalar map
        $\psi (t) := (b + \beta t^{1/\xi_2})^{\xi_2}$ where $t > 0$.
    Since $\xi_2 <0$,
    it is clear that the scalar function $\psi$ 
    is continuous, strictly increasing and strictly concave on $\RR_{++}$ (cf.
    section~\ref{sss:theta<0}).
    The first part of claim is immediate from the monotonicity and concavity of $\psi$,
    along with monotonicity and linearity of the integral.

    For the remaining part, fix $h$ in $bm(\SS \times \ZZ \times \Theta)_{++}$.
    Borel measurability of $((s,z), y) \mapsto \tilde H ((s,z),y,h)$ is obvious.
    Now fix a function $h$ satisfying that 
    the map $h(\cdot, \cdot, \theta) \colon \SS \times \ZZ \to \RR_{++}$ is continuous, 
    for every $\theta \in \Theta$.  
    By virtue of the continuity imposed on the distribution $\mu(\cdot,
    \cdot)$ and the finiteness of $\Theta$,
    the map 
    $\SS \times \ZZ \ni (y,z) \mapsto \int h(y,z,\theta) \mu (z, \diff \theta) \in \RR_{++}$ 
    is continuous on $\SS \times \ZZ$.
    It then follows from the continuity imposed on $r$ 
    and the continuity of $\psi$
    that $((s,z),y) \mapsto \tilde H ((s,z),y,h)$ is continuous on $\GG$.
\end{proof}

\begin{lemma}
    \label{l:q-sap}
    If $\xi_1 \in (0,1)$ and $\xi_2 < 0$,
    then the state-action aggregator $H$ defined in~\eqref{eq:q-def-sap} is isotone 
    and concave in its third argument on $bm(\SS \times \ZZ)_{++}$. 

    In addition, the map $((s,z),y) \mapsto H((s,z),y, v)$ is Borel measurable on $\GG$
    whenever $v \in bm(\SS \times \ZZ)_{++}$
    and continuous on $\GG$
    whenever $v \in bc(\SS \times \ZZ)_{++}$.
\end{lemma}

\begin{proof}[Proof of lemma~\ref{l:q-sap}]
    Since the aggregator $H$ is a composition of $\tilde H$ and $R$, 
    by lemmas~\ref{l:r-theta} to \ref{l:tq},
    the isotonicity, Borel measurability and continuity of $H$ immediately follow from 
    those of $\tilde H$ and $R$.

    It only remains to show the concavity of $H$.
    To see this,
    fix $((s,z),y) \in \GG$,
    $\lambda \in [0,1]$ and $w, w'$ in $bm(\SS \times \ZZ)_{++}$.
    For any given $\theta \in \Theta$, 
    by concavity of $R_\theta$ and convexity of $bm(\SS \times \ZZ)_{++}$, 
    we have 
    \begin{align*}
        R_\theta [ \lambda w + (1-\lambda) w' ] (y,z)
        \geq \lambda R_\theta w(y,z) + (1-\lambda) R_\theta w'(y,z);
    \end{align*}
    that is, for each $(y,z,\theta) \in \SS \times \ZZ \times \Theta$,
    \begin{align*}
        R[\lambda w + (1-\lambda) w'](y,z,\theta) 
        \geq \lambda R w(y,z, \theta) + (1-\lambda) R w'(y,z, \theta).
    \end{align*}
    In operator notation, 
    this translates to $R[ \lambda w + (1-\lambda) w'] \geq \lambda R w + (1-\lambda) R w'$.

    Observe that due to isotonicity and concavity of $\tilde H$, we now obtain
    \begin{align*}
        H((s,z),y,\lambda w + (1-\lambda) w')
        & = \tilde H ((s,z),y, R[\lambda w + (1-\lambda) w']) \\
        & \geq \tilde H((s,z),y, \lambda R w + (1-\lambda) R w') \\
        & \geq \lambda \tilde H((s,z),y, Rw) + (1-\lambda) \tilde H((s,z),y, Rw') \\
        & = \lambda H((s,z),y, w) + (1-\lambda) H((s,z),y, w'),
    \end{align*}
    where the first and last equalities follow immediately from 
    the definition of $H$ in~\eqref{eq:H-def},
    while the first and second inequalities follow from isotonicity and concavity of $\tilde H$, respectively.  This completes the proof.
\end{proof}

Analogously, the state-action aggregator $H$ defined in~\eqref{eq:q-sap-rho=1}
can be expressed as
\begin{equation*}
    H((s,z), y ,\hat v)
    = \tilde H((s,z), y , R \hat v),
\end{equation*}
with the operator $R$ defined as above, but 
\begin{equation}
    \label{eq:tH-sap=1}
    \tilde H((s,z), y , h) 
    := \exp \left( (1-\eta) \left\{ r(s,y,z)
      + \frac{\beta}{1-\eta} \ln \left[ \int h(y,z, \theta) \mu(z, \diff \theta) 
      \right] \right\} \right)
\end{equation}
for all $((s,z), y) \in \GG$ and $h \in bm(\SS \times \ZZ \times \Theta)_{++}$.

Observe that the formula of $\tilde H$ defined above
is almost identical to that of $H$ defined in~\eqref{eq:q-rsp}.
In this connection, 
recalling the results associated with $H$ in section~\ref{ss:riskssp-app},
it is easy to see that 
\begin{lemma}
    \label{l:tq-explog}
    If $\eta > 1$,
    then $\tilde H$ defined in \eqref{eq:tH-sap=1} is isotone and concave in its third argument
    on $bm(\SS \times \ZZ \times \Theta)_{++}$.

    In addition, the map $((s,z), y) \mapsto \tilde H ((s,z),y,h)$ is Borel measurable on $\GG$
    whenever $h \in bm(\SS \times \ZZ \times \Theta)_{++}$,
    and continuous on $\GG$
    whenever the map $h(\cdot, \cdot, \theta) \colon \SS \times \ZZ \to \RR_{++}$ is continuous,
    for each $\theta \in \Theta$.
\end{lemma}

\begin{proof}[Proof of lemma~\ref{l:tq-explog}]
    Analogous to the proof of lemma~\ref{l:tq},
    for fixed $b \in \RR$,
    we consider the scalar map
    \begin{equation*}
        \label{eq:psi-sap=1}
        \psi (t) 
        := \exp \left[ (1-\eta) \left( b + \frac{\beta}{1-\eta} \ln t \right) \right]
        \qquad (t > 0).
    \end{equation*}
    It is clear that this scalar function $\psi$ is continuous, strictly increasing and 
    strictly concave on $\RR_{++}$.\footnote{
        For more details of the relevant results of such $\psi$, please refer to
        section~\ref{ss:riskssp-app}.
        }
    As a consequence, the remaining proof of lemma~\ref{l:tq-explog} is identical to
    that of lemma~\ref{l:tq}, and thus omitted here.
\end{proof}

\begin{lemma}
    \label{l:q-sap-rho=1}
    If $\xi_1 \in (0,1)$ and $\eta > 1$,
    then the state-action aggregator $H$ defined in~\eqref{eq:q-sap-rho=1} is isotone 
    and concave in its third argument on $bm(\SS \times \ZZ)_{++}$. 

    In addition, the map $((s,z),y) \mapsto H((s,z),y, v)$ is Borel measurable on $\GG$
    whenever $v \in bm(\SS \times \ZZ)_{++}$
    and continuous on $\GG$
    whenever $v \in bc(\SS \times \ZZ)_{++}$.
\end{lemma}

\begin{proof}[Proof of lemma~\ref{l:q-sap-rho=1}]
    Invoking lemmas~\ref{l:r-theta}, \ref{l:R-mer} and \ref{l:tq-explog},
    the proof is identical to that of lemma~\ref{l:q-sap},
    and hence is omitted.
\end{proof}

\subsection{Proofs for section~\ref{ss:narfram}}

\label{ss:prf-nf}

\begin{proof}[Proof of lemma~\ref{l:slu-nf}]
    Let constants $m$ and $M$ be as defined in~\eqref{eq:mm}.
    As $B$ is continuous on a compact set, there exists a finite constants
    \begin{equation*}
        \label{eq:ll}
        l :=  \min_{((s,z),y) \in \GG}  B(s, y, z) 
        \; \text{ and } \;
        L :=  \max_{((s,z),y) \in \GG} B(s, y, z) .
    \end{equation*}
    \textbf{Case I : $\rho<1<\gamma$.}
    To show condition (SL) of lemma~\ref{l:slu-nf},
    we first claim that there exists a positive constant function $w_1$
    such that
    for fixed $((s,z),y) \in \GG$, we have
    \begin{align}
        \label{eq:strictl-nf}
        H((s,z), y, w_1) 
        = \left\{
        r(s, y, z) + \beta \left[ w_1^{\frac{1}{1-\gamma}} + B(s,y,z)
        \right]^{1-\rho} \right\}^\theta \nonumber 
        &\geq \left\{ M + \beta \left[ w_1^{\frac{1}{1-\gamma}} + L \right]^{1-\rho} \right\}^\theta \nonumber \\
        & > w_1 (s,z).
    \end{align}
    Evidently, the uniformly strict inequality~\eqref{eq:strictl-nf} implies 
    that such a positive constant function $w_1$ satisfies condition (SL).

    To prove our claim that 
    there exists a positive constant function $w_1$ satisfying~\eqref{eq:strictl-nf}, 
    we note that, since $0< 1- \rho < 1$ and $\theta < 0$, the following equivalence relation holds
    \begin{align*}
        \left\{ M + \beta \left[ w_1^{\frac{1}{1-\gamma}} + L \right]^{1-\rho} \right\}^\theta > w_1
        &\iff
        \left( \frac{w_1^{\frac{1}{\theta}}-M}{\beta} \right)^{\frac{1}{1-\rho}} - w_1^{\frac{1}{1-\gamma}} - L >0.
    \end{align*}
    Let $d := w_1^{\frac{1}{1-\gamma}}$ and set
    \begin{equation*}
        \varphi (d)
        := \left( \frac{d^{1-\rho}-M}{\beta} \right)^{\frac{1}{1-\rho}} - d - L
        \qquad (d > 0),
    \end{equation*}
    Showing that ~\eqref{eq:strictl-nf} holds
    is equivalent to showing that there exits a positive constant $d^*$ such that $\varphi(d^*) > 0$.
    To show that the latter holds true,
    one can verify that both the first and the second derivatives of $\varphi$ 
    on the interval $(\underline d, \, \infty) \subset \mathbbm{R}_{++}$ 
    are positive, where $\underline d := [M /(1-\beta^{1/\rho})
    ]^{1/(1-\rho)}$.  (We have $\underline d > 0$, since $M \geq m >0$.)
    Hence $\varphi$ is concave upward on $(\underline d, \, \infty)$. 
    This means that $\varphi(d)$ goes to $\infty$, 
    as $d \to \infty$, which in turn implies 
    that there exists a positive constant $d^* > \underline d$ 
    such that $\varphi(d^*) > 0$.
    Letting $w_1 := (d^*)^{1-\gamma}$ finishes the proof of condition (SL).

    Regarding condition (U) of lemma~\ref{l:slu-nf},
    we claim first that there is a positive constant function $w_2$ such that
    for fixed $((s,z),y) \in \GG$, we have 
    \begin{align}
        \label{eq:upper-nf}
        H((s,z), y, w_2) 
        = \left\{ r(s, y, z) + \beta \left[ w_2^{\frac{1}{1-\gamma}} + B(s,y,z)
        \right]^{1-\rho} \right\}^\theta \nonumber 
        &\leq \left\{ m + \beta \left[ w_2^{\frac{1}{1-\gamma}} + l \right]^{1-\rho} \right\}^\theta \nonumber \\
        &\leq w_2 (s,z).
    \end{align}
    Evidently, to show the existence of an upper solution $w_2$, it is sufficient to show that 
    there exists a positive constant function $w_2$ satisfying~\eqref{eq:upper-nf}.
    Further, after some rearrangement, 
    we note showing that ~\eqref{eq:upper-nf}
    holds is equivalent to showing that 
    
    \begin{equation*}
        \left( \frac{w_2^{\frac{1}{\theta}} - m}{\beta} \right)^{\frac{1}{1-\rho}}
        - w_2^{\frac{1}{1-\gamma}} - l \leq 0.
    \end{equation*}

    Let $w_2 := [ m/(1-\beta^{1/\rho}) ]^\theta$. 
    Then the left-hand side of the preceding inequality equals 
    \begin{align*}
        \left( \frac{ \frac{\beta^{1/\rho}}{1 - \beta^{1/\rho}} m}{\beta} \right)^{\frac{1}{1-\rho}}
        - \left( \frac{m}{1-\beta^{\frac{1}{\rho}}}\right)^{\frac{1}{1-\rho}} - l 
        = \left[ \beta^{\frac{1}{\rho} } - 1\right] \left( \frac{m}{1-\beta^{\frac{1}{\rho}}}\right)^{\frac{1}{1-\rho}}
        - l .
    \end{align*}
    Since $\beta \in (0,1)$ and $\rho \in (0,1)$, we have $\beta^{1/\rho} -1 <0$.
    Further, it follows from $m > 0$ and $l \geq 0$ that
    the right-hand side of the above equality is negative.
    This, in turn, implies that for $w_2$ defined above, \eqref{eq:upper-nf} is satisfied,
    which proves condition (U).

    To see that $w_1 < w_2$, observe that 
    $0 < \underline d < d^*$ and $1-\gamma < 0$ imply 
    $0 < w_1 \equiv (d^*)^{1-\gamma} < (\underline d)^{1-\gamma}$. 
    In addition, since $ m \leq  M$ and $\theta < 0$, we have 
    $(\underline d)^{1-\gamma} \equiv [ M/(1-\beta^{1/\rho}) ]^\theta 
    \leq [ m/(1-\beta^{1/\rho}) ]^\theta \equiv w_2$.
    We can now conclude that $w_1 < w_2$, as desired.

    \textbf{Case II : $1<\rho<\gamma$.}
    For this case, the proof is similar. 
    Regarding condition (SL) of lemma~\ref{l:slu-nf},
    we claim first that there exists a positive constant function $w_1$ such that
    for fixed $((s,z),y) \in \GG$, we have
    \begin{align}
        \label{eq:strictl-nf-rho>1}
        H((s,z), y, w_1) 
        = \left\{
        r(s, y, z) + \beta \left[ w_1^{\frac{1}{1-\gamma}} + B(s,y,z)
        \right]^{1-\rho} \right\}^\theta \nonumber 
        &\geq \left\{ m + \beta \left[ w_1^{\frac{1}{1-\gamma}} + L \right]^{1-\rho} \right\}^\theta \nonumber \\
        & > w_1 (s,z).
    \end{align}
    The uniformly strict inequality~\eqref{eq:strictl-nf-rho>1} implies that
    $w_1$ satisfies condition (SL).

    To show that there exists a positive constant function $w_1$
    satisfying~\eqref{eq:strictl-nf-rho>1}, 
    we note that, since $1- \rho < 0$ and $\theta > 1$,
    the following equivalence relation holds
    \begin{align*}
        \left\{ m + \beta \left[ w_1^{\frac{1}{1-\gamma}} + L \right]^{1-\rho} \right\}^\theta > w_1
        &\iff
        \left( \frac{w_1^{\frac{1}{\theta}}-m}{\beta} \right)^{\frac{1}{1-\rho}} - w_1^{\frac{1}{1-\gamma}} - L >0.
    \end{align*}
    Let $d \equiv w_1^{\frac{1}{1-\gamma}}$ and set
    \begin{equation*}
        \phi (d)
        := \left( \frac{d^{1-\rho}-m}{\beta} \right)^{\frac{1}{1-\rho}} - d - L
        \qquad (d > 0),
    \end{equation*}
    Showing that ~\eqref{eq:strictl-nf-rho>1} holds
    is equivalent to showing that there exits a positive constant $d^*$ such that $\phi(d^*) > 0$.
    To show the latter holds,
    one can check that both the first and second derivatives of $\phi$ 
    on the interval $(\underline d, \, m^{1/(1-\rho)}) \subset \mathbbm{R}_{++}$ 
    are positive, where $ \underline d := [m /(1-\beta^{1/\rho}) ]^{1/(1-\rho)}$.
    Hence, the graph of $\phi$ 
    on $(\underline d, \, m^{1/(1-\rho)})$ is concave upward. 
    Hence $\phi(d)$ approaches $+\infty$
    as $d$ approaches $m^{1/(1-\rho)}$.
    It follows that there exists a positive constant $d^* \in (\underline d, \, m^{1/(1-\rho)})$ 
    satisfying $\phi(d^*) > 0$.
    Finally, for such $d^*$, 
    letting $w_1 \equiv (d^*)^{1-\gamma}$ finishes the proof of condition (SL).

    Next, to show condition (U), we claim first that
    there is a positive constant function $w_2$ such that
    for fixed $((s,z),y) \in \GG$,
    we have 
    \begin{align}
        \label{eq:upper-nf-rho>1}
        H((s,z), y, w_2) 
        = \left\{ r(s, y, z) + \beta \left[ w_2^{\frac{1}{1-\gamma}} + B(s,y,z)
        \right]^{1-\rho} \right\}^\theta \nonumber 
        &\leq \left\{ M + \beta \left[ w_2^{\frac{1}{1-\gamma}} + l \right]^{1-\rho} \right\}^\theta \nonumber \\
        &\leq w_2 (s,z).
    \end{align}
    To show the existence of an upper solution $w_2$, it suffices to show that 
    there exists a positive constant function $w_2$ satisfying~\eqref{eq:upper-nf-rho>1},
    or equivalently,
    \begin{equation*}
        \left( \frac{w_2^{\frac{1}{\theta}} - M}{\beta} \right)^{\frac{1}{1-\rho}}
        - w_2^{\frac{1}{1-\gamma}} - l \leq 0.
    \end{equation*}

    Let $w_2 := [M /(1-\beta^{1/\rho}) ]^\theta$. The left-hand side of the above inequality is equal to
    \begin{align*}
        \left( \frac{ \frac{\beta^{1/\rho}}{1 - \beta^{1/\rho}} M}{\beta} \right)^{\frac{1}{1-\rho}}
        - \left( \frac{M}{1-\beta^{\frac{1}{\rho}}}\right)^{\frac{1}{1-\rho}} - l  
        &= \left[ \beta^{\frac{1}{\rho} } - 1\right] \left( \frac{M}{1-\beta^{\frac{1}{\rho}}}\right)^{\frac{1}{1-\rho}}
        - l 
        < 0.
    \end{align*}
    This in turn implies that for such $w_2$ defined above, 
    \eqref{eq:upper-nf-rho>1} is naturally satisfied,
    which is what we needed to show for condition (U).

    Our choices of $w_1$ and $w_2$ satisfy $w_1 < w_2$.
    To see that this is so, observe that 
    $w_1 \equiv (d^*)^{1-\gamma} < (\underline d)^{1-\gamma} \equiv [m /(1-\beta^{1/\rho}) ]^\theta$.
    Furthermore, it follows from $\theta > 1$ and $m \leq M$
    that $[m /(1-\beta^{1/\rho}) ]^\theta \leq [M /(1-\beta^{1/\rho}) ]^\theta \equiv w_2$,
    from which we conclude that $w_1 < w_2$, as was to be shown. 
\end{proof}

\subsection{Proofs in section~\ref{s:ubdd}}

\label{ss:prf-ubdd}

Recalling the definition of the weight function $\ell$ in section~\ref{s:ubdd},
the finite $\ell$-norm turns the real normed vector space 
$b_\ell m\XX := \setntn{f \in m\XX}{\text{$f$ is $\ell$-bounded}}$
into a real Banach space.\footnote{
    The \emph{$\ell$-norm} of $f$
    is defined by $\| f \|_\ell := \sup_{(s,z) \in \SS \times \ZZ} \{ |f(s,z)| / \ell(s,z) \}$.
    It is worth noting that when the weight function $\ell$ is bounded,
    the $\ell$-norm $\| \cdot \|_\ell$ and the supremum norm $\| \cdot \|$ are equivalent.
    Therefore, the weighted supremum norms become relevant when $\ell$ is unbounded.
    }
Recalling the definition of $H$ and the construction of bracketing functions $w_1$ and $w_2$
in section~\ref{s:ubdd},
we obtain the following results.

\begin{lemma}
    \label{l:ic-ext}
    If assumption~\ref{a:app-ext} holds,
    then the state-action aggregator $H$ defined in~\eqref{eq:q-ez}
    is isotone and concave in its third argument on $\vV$.
\end{lemma}

\begin{proof}
    The proof is essentially the same as the isotonicity and value-concavity arguments on
    $H$ provided in section~\ref{sss:theta>1}.
\end{proof}

\begin{lemma}
    \label{l:uls-ext}
    If assumption~\ref{a:app-ext} holds,
    then the state-action aggregator $H$ defined in~\eqref{eq:q-ez} possesses
    a strict lower solution $w_1$ and an upper solution $w_2$ in the sense that
    \begin{enumerate}
        \item[(SL)] there exists an $\epsilon > 0$ such that 
            $H((s,z), y, w_1) \geq w_1(s,z) + \epsilon \kappa (s,z)^\theta$ 
            for all $((s,z), y) \in \GG$.
        \item[(U)] $H((s,z), y, w_2) \leq w_2(s,z)$ for all $((s,z), y) \in \GG$.
    \end{enumerate}
\end{lemma}

\begin{proof}[Proof of lemma~\ref{l:uls-ext}]
    Observe that,
    for fixed $((s,z),y) \in \GG$, we have
    \begin{align}
        \label{eq:ineq-ext}
        H((s,z),y, w_1)
        &= \left\{ r(s,y,z)
        + \beta  \left[ \int (L - \delta )^\theta \cdot \kappa(y,z') P(z, \diff z')
        \right]^{1/\theta} \right\}^\theta \nonumber \\
        & \geq \left\{ L \kappa (s,z)
        + \beta  \left[ \int (L - \delta )^\theta \cdot \kappa(y,z') P(z, \diff z')
        \right]^{1/\theta} \right\}^\theta \nonumber \\
        & \geq \left\{ L \kappa (s,z)
        + \beta (L - \delta )  \left[ d \kappa (s,z)
        \right]^{1/\theta} \right\}^\theta \\
        & \geq \left\{ \left[ L 
        + \beta (L - \delta )  d^{1/\theta} \right] 
        \cdot \kappa (s,z)^{1/\theta} \right\}^\theta, \nonumber
    \end{align}
    where the first and second inequalities immediately 
    follow from~\eqref{eq:a-ext-if} and ~\eqref{eq:a-ext-ik} in assumption~\ref{a:app-ext}, respectively,
    while the last one from the fact that $\kappa^{1/\theta} \leq \kappa$.
    Further, with some rearranging, we obtain
    \begin{equation*}
        H((s,z),y, w_1)
        \geq \left[ L -\delta
        + \beta L d^{1/\theta} + \delta (1-\beta d^{1/\theta})
        \right]^\theta \kappa (s,z)
        > (L - \delta)^\theta \kappa(s,z),
    \end{equation*}
    and the last term is equal to $w_1(s,z)$.
    Similarly, we have
    \begin{align*}
        H((s,z),y, w_2)
        &= \left\{ r(s,y,z)
        + \beta \left( \frac{M}{1-\beta c^{1/\theta}} 
        \right) \left[ \int \kappa(y,z')^\theta P(z, \diff z')
        \right]^{1/\theta} \right\}^\theta \\
        & \leq \left\{ M \kappa(s,z)
        + \beta \left( \frac{M}{1-\beta c^{1/\theta}} 
        \right) \left[ c \kappa(s,z)^\theta
        \right]^{1/\theta} \right\}^\theta \\
        & = \left[ \frac{M}{1-\beta c^{1/\theta}} \right]^\theta
        \kappa(s,z)^\theta
        = w_2(s,z)
    \end{align*}
    where the inequality follows from~\eqref{eq:a-ext-sf} and \eqref{eq:a-ext-sk}
    in assumption~\ref{a:app-ext}.
    Hence condition (U) of lemma~\ref{l:uls-ext} is satisfied.

    So far we only show $w_1$ is a lower solution of $H$. 
    It remains to prove that it is a strict lower solution. 
    Observe from~\eqref{eq:ineq-ext} that to show condition (SL), it is sufficient to show
    that there exists an $\epsilon >0$ such that 
    \begin{equation}
        \label{eq:sls-ext}
        \left\{ L \kappa (s,z)
        + \beta (L - \delta )  \left[ d \kappa (s,z)
        \right]^{1/\theta} \right\}^\theta
        \geq w_1(s,z) + \epsilon \kappa(s,z)^\theta
    \end{equation}
    for all $(s,z) \in \SS \times \ZZ$.
    To this end, for fixed $(s,z) \in \SS \times \ZZ$, consider
    \begin{align*}
        &\frac{\left\{ L \kappa (s,z)
        + \beta (L - \delta )  \left[ d \kappa (s,z)
        \right]^{1/\theta} \right\}^\theta - w_1(s,z)}{\kappa(s,z)^\theta} \\
        & = \left\{ L 
        + \beta (L - \delta )  d^{1/\theta} 
        \cdot \kappa (s,z)^{1/\theta - 1} \right\}^\theta 
        - (L - \delta)^\theta \kappa(s,z)^{1-\theta} \\
        & \geq L^\theta - (L - \delta)^\theta \kappa(s,z)^{1-\theta}
        \geq L^\theta - (L - \delta)^\theta > 0
    \end{align*}
    where the first and second inequalities follow from the facts that 
    $\kappa^{1/\theta -1} \geq 0$ and that $\kappa^{1-\theta} \leq 1$, respectively.
    Hence, 
    condition~\eqref{eq:sls-ext} holds 
    when we take $\epsilon := L^\theta - (L - \delta)^\theta$, 
    which is what we needed to show for condition (SL).
\end{proof}

\begin{lemma}
    \label{l:conti-ext}
    If assumption~\ref{a:app-ext} holds, then the map
    $((s,z),y) \mapsto H((s,z),y, \hat v)$ is continuous on $\GG$ whenever $\hat v \in \cC$.
\end{lemma}

\begin{proof}[Proof of lemma~\ref{l:conti-ext}]
    To see that this is so, pick any $\hat v \in \cC$.  By
    assumption~\ref{a:app-ext}, and by lemma 12.2.20 in
    \cite{stachurski2009economic}, we know that $(y,z) \mapsto \int \hat
    v(y,z') P(z, \diff z')$ is continuous on $\SS \times \ZZ$.
    It then follows from the continuity of $r$ that 
    the map $((s,z),y) \mapsto H((s,z),y, \hat v)$ is continuous on $\GG$,
    as was to be shown.
\end{proof}

Recall that the $\sigma$-value operator $T_\sigma$ on $\vV$ is defined by
\begin{equation*}
    T_\sigma \hat v(s,z) 
    = H((s,z), \sigma(s,z), \hat v)
    = \left\{ r_\sigma(s,z)
      + \beta  \left[ \int \hat v(\sigma(s,z),z') P(z, \diff z')
      \right]^{1/\theta} \right\}^\theta
\end{equation*}
for all $(s,z) \in \SS \times \ZZ$ and $\hat v \in \vV$,
where $r_\sigma(s,z):= r(s,\sigma(s,z),z)$.

\begin{lemma}
    \label{l:ts-ext}
    If assumption~\ref{a:app-ext} holds,
    then, for each $\sigma \in \Sigma$,
    the operator $T_\sigma$ is asymptotically stable on $\vV$.
\end{lemma}

\begin{proof}[Proof of lemma~\ref{l:ts-ext}]
    First, we show that $T_\sigma$ is a self-map on $\vV$.
    Fix $\hat v $ in $\vV$.
    Together with continuity of $r$, measurabilities of $\hat v$ and $\sigma$
    imply that $T_\sigma \hat v $ is Borel measurable on $\SS \times \ZZ$.
    In addition, since $w_1 \leq \hat v$, making use of isotonicity of $H$ 
    (as shown in lemma~\ref{l:ic-ext}), we have 
    $w_1(s,z) \leq H((s,z), \sigma(s,z), w_1) \leq H((s,z), \sigma(s,z), \hat v)$
    for all $(s,z) \in \SS \times \ZZ$, 
    which in turn implies that $w_1 \leq T_\sigma \hat v$. 
    A similar argument gives $T_\sigma \hat v \leq w_2$. 
    Therefore, $T_\sigma \hat v \in \vV$, as was to be shown.

    Now, invoking lemmas~\ref{l:ic-ext} and \ref{l:uls-ext}, 
    theorem 3.1 of \cite{du1990} applies and implies the stated result.
\end{proof}

\begin{proof}[Proof of proposition~\ref{p:regp-ccv-ext}]
    This follows immediately from lemma~\ref{l:ts-ext}.
\end{proof}

Given $\hat v \in \vV$, a policy $\sigma$ in $\Sigma$
will be called \emph{$\hat v$-greedy} if 
\begin{equation*}
    \sigma(s,z) 
    \in \argmin_{y \in \Gamma(s,z)} H((s,z), y, \hat v)
    \: \text{for all } (s,z) \in \SS \times \ZZ.
\end{equation*}

\begin{lemma}
    If $\hat v \in \cC$,
    then there exists at least one $\hat v$-greedy policy.
\end{lemma}

\begin{proof}
    The proof is essentially identical to that of lemma~\ref{l:egp},
    and hence omitted.
\end{proof}

\begin{lemma}
    \label{l:s-ext}
    If assumption~\ref{a:app-ext} holds, then $S$ is asymptotically stable on $\cC$.
\end{lemma}

\begin{proof}[Proof of lemma~\ref{l:s-ext}]
    By virtue of lemma~\ref{l:conti-ext},
    it follows from Berge's theorem of the minimum that,
    when $\hat v$ is in $\cC$, we have
    \begin{equation*}
        S \hat v (s,z)
        = \min_{y \in \Gamma(s,z)} H((s,z), y, \hat v)
        = \min_{y \in \Gamma(s,z)}
      \left\{ r(s,y,z)
      + \beta  \left[ \int \hat v(y,z') P(z, \diff z')
      \right]^{1/\theta} \right\}^\theta
    \end{equation*}
    and $S \hat v$ is an element of $\cC$.

    In order to apply \citeauthor{du1990}'s theorem to the Bellman operator $S$, 
    it suffices to show that
    \begin{enumerate}
        \item[(i)] $S$ is isotone and concave on $\cC$, and
        \item[(ii)] $S w_1 \gg w_1$ and $S w_2 \leq w_2$.\footnote{
            The symbol $\gg$ denotes the strong partial order in 
            the Banach space $b_{\kappa^\theta} c(\SS \times \ZZ)$ 
            of all $\kappa^\theta$-bounded continuous functions on $\SS \times \ZZ$, 
            in the sense that $w \gg v$ means $w- v$ lies in the interior of 
            $b_{\kappa^\theta} c(\SS \times \ZZ)_+$. 
            For more details, please refer to \cite{zhang2012variational}.
            }
    \end{enumerate}
    Regarding part (i), making use of the result of lemma~\ref{l:ic-ext},
    the proof is essentially identical to the dual proof of lemma~\ref{l:st}.
    In addition, making use of the result of lemma~\ref{l:uls-ext},
    the proof of part (ii) is also essentially identical to 
    the dual proof of lemma~\ref{l:st}.
\end{proof}

\begin{proof}[Proof of theorem~\ref{t:bkcv-ext}]
    By lemmas~\ref{l:ts-ext} to \ref{l:s-ext},
    applying the dual proof of theorem~\ref{t:ccv-app} yields the stated results
    in theorem~\ref{t:bkcv-ext}.
\end{proof}

\bibliographystyle{ecta}

\bibliography{rcdp}

\end{document}